\documentclass[12pt]{article}
\usepackage{amsmath}
\usepackage{graphicx,psfrag,epsf}
\usepackage{enumerate}
\usepackage{natbib}
\usepackage{url} 
\newcommand{\blind}{0}

\usepackage{xcolor}
\usepackage{amsthm}
\usepackage[colorlinks=true, citecolor = blue]{hyperref}
\usepackage{amssymb}
\usepackage{bbm, dsfont}
\usepackage[justification=justified, belowskip=0.2cm, skip=0.6\baselineskip, labelsep=period]{caption}
\usepackage{multirow}
\usepackage{longtable}
\usepackage{booktabs} 
\usepackage{eurosym}

\newcommand{\R}{\mathbb{R}}
\newcommand{\one}{\mathbbm{1}}
\newcommand{\E}{\mathbb{E}}

\newcommand{\IN}{\mathbb{N}}
\newcommand{\I}{\mathcal{I}}

\renewcommand{\thefigure}{\thesection.\arabic{figure}}
\numberwithin{equation}{section}

\newcommand{\G}{\mathbb{G}}

\newcommand{\X}{{X}}
\newcommand{\xo}{x}
\newcommand{\x}{{x}}

\newcommand{\T}{\mathcal{T}}

\newcommand{\Supp}{\mathcal{S}}
\newcommand{\B}{\mathcal{B}(\T,\Theta)}

\newcommand{\te}{\theta}

\newcommand{\dint}{ \displaystyle\int}
\newcommand{\bt}{ \,\mid \,}

\DeclareMathOperator*{\argmin}{argmin}
\newtheorem{theorem}{Theorem}  
\newtheorem{lemma}{Lemma}
\newtheorem{corollary}{Corollary}  
\newtheorem{assumption}{Assumption}

\begin{document}

\bibliographystyle{ecta}

\def\spacingset#1{\renewcommand{\baselinestretch}%
{#1}\small\normalsize} \spacingset{1}

\thispagestyle{empty}

\if0\blind
{
  \title{\bf Flexible Specification Testing in Quantile Regression Models}
  \author{Tim Kutzker\thanks{Tim Kutzker (\textit{corresponding author}, tim.kutzker@hu-berlin.de) and Nadja Klein (nadja.klein@hu-berlin.de) acknowledge support by the Deutsche Forschungsgemeinschaft (DFG, German research foundation) through the Emmy Noether grant KL 3037/1-1. We thank Matteo Fasiolo for helpful comments on the \texttt{qgam} \citep{fasiolo2020fast} package in R.  The authors acknowledge computational facilities by the Regional Computing Center of the University
of Cologne via the DFG-funded High Performance Computing system CHEOPS (grant INST 216/512/1FUGG).}\ \ and 
    Nadja Klein\footnotemark[1]\\
    Emmy Noether Research Group in Statistics and Data Science,\\ Humboldt-Universität zu Berlin, Germany \\
    and\\
    Dominik Wied\thanks{dwied@uni-koeln.de} \\
    Econometrics and Statistics, Universität zu Köln, 
    Germany}
    \date{}
  \maketitle\thispagestyle{empty}
} \fi

\if1\blind
{
  \bigskip
  \bigskip
  \bigskip
  \begin{center}
    {\LARGE\bf Flexible Specification Testing in Quantile Regression Models}
\end{center}
  \medskip
} \fi
\bigskip
\begin{abstract}
\noindent We propose three novel consistent specification tests for quantile regression models which generalize former tests in three ways. First, we allow the covariate effects to be quantile-dependent and nonlinear. Second, we allow  parameterizing the conditional quantile functions by appropriate basis functions, rather than parametrically. We are hence able to test for functional forms beyond linearity, while retaining the linear effects as special cases. In both cases, the induced class of conditional distribution functions is tested with a Cram\'{e}r-von Mises type test statistic for which we derive the theoretical limit distribution and propose a  bootstrap method. Third, to increase the power of the tests, we further suggest a  modified test statistic. 
We highlight the merits of our tests in a detailed MC study and two real data examples. Our first application to conditional income distributions in Germany  indicates that there are not only still significant differences between East and West but also across the quantiles of the conditional income distributions, when conditioning on age and year. The second application to data from the Australian national electricity market reveals the importance of using interaction effects for modelling the highly skewed and heavy-tailed distributions of energy prices conditional on day, time of day and demand.
\end{abstract}

\noindent%
\vfill
\setcounter{page}{1}
\newpage
\spacingset{1.8} 

\setcounter{page}{1}

\section{Introduction}
Hypothesis testing plays a central role in many research areas. A necessary prerequisite for the statistical validity of the decisions to be made is the correct specification of the underlying model. Specification tests can be used to validate the correctness of theoretical assumptions. For linear regression, a whole range of specification tests are available for both, parametric and non-parametric approaches. In general, testing misspecification in linear ordinary least squares (OLS) models is well understood and developed. 
\\
In parametric models, e.g.~, \citet{bierens1990consistent} showed that any conditional moment test of functional form of nonlinear regression models can be converted into a consistent chi-squared test that is consistent against all deviations from the null hypothesis. \citet{hardle1993comparing} suggested a wild bootstrap procedure for regression fits in order to decide whether a parametric model could be justified, while \citet{stute1997nonparametric} proposed a more general method for testing the goodness of fit of a parametric regression model. For the non-parametric case, amongst others, \citet{gozalo1993consistent} proposed a general framework for specification testing of the regression function in a non-parametric smoothing estimation context and \citet{stute1998bootstrap} suggested a goodness of fit test using a wild bootstrap procedure that checks whether a function belongs to a certain  class. 
\\
However, OLS estimates are sensitive to outliers and draw only a part of the whole picture since they only model the mean. In contrast, quantile regression  provides more robust estimates and allows a more comprehensive picture of the entire conditional distribution. Due to these advantages, quantile regression has become increasingly popular since the seminal article by \citet{koenker1978regression}. However, post-estimation inference procedures for quantile regression models essentially depend on the validity of the underlying parametric functional form for the quantiles considered \citep[][]{angrist}. For example, assuming the same fixed linear relationship between covariates for all quantiles is the connecting element of the Machado-Mata (M-M) decomposition in order to describe wage inequalities \citep[][]{machadomata} and the Khmaladze transformation \citep[][]{koenkerxiao}. Thus, testing the validity of the imposed structure remains one of the key taks associated with challenges for valid posterior inference.  
\\
In a parametric framework, one of the first specification tests for linear location shift and location-scale shift quantile models with i.i.d.~data is the test by \citet{koenkerxiao}. Shortly thereafter, \citet{cher:2002} proposes a resampling test procedure that avoids the estimation of additional objects, such as the score function, while building on the principles stated in \citet{koenkerxiao}. However, these two tests  do not test the validity of the quantile regression model itself. \citet{escanciano2010specification}  and \citet{escanciano2010specification1} both tested the validity of the null hypothesis that a conditional quantile restriction is valid over a range of quantiles. \citet{rothewied:2013} proposed a specification test for a larger class of models, including quantile regression models. This principle was extended to dynamic models by \citet{TrosterWied2021}. In case of non-parametric instrumental quantile regression, \citet{breunig2019specification} develops a methodology for testing the hypothesis whether the instrumental quantile regression model is correctly specified. However, all models have in common that they require linearity in the regressors. 
\\
Since such a linearity assumption considerably limits the number of possible models and hence the hypothesis space, there have recently been successful attempts to weaken the linearity assumption for quantile estimation and inference with independently and identically distributed (i.i.d.) data.  In this context, more general parametric quantile models have been developed that, amongst others, include works by \citet{hallin2009local} suggesting an estimator for local linear spatial quantile regression and  \citet{guerre2012uniform} investigating the Bahadur representation of a local polynomial estimator of the conditional quantile function (qf) and its derivatives. But also non-parametric approaches for estimating conditional qfs have attracted much attention. \citet{li2008nonparametric} proposed a non-parametric conditional cumulative distribution function (cdf) kernel estimator along with an associated non-parametric conditional quantile estimator.  \citet{belloni2019conditional} developed  non-parametric quantile regression  for performing inference on the entire conditional qf and its linear functionals and \citet{qu2015nonparametric} presented estimators for non-parametrically specified conditional quantile processes that are based on local linear regressions. \citet{li2020nonparametric} investigated the problem of non-parametrically estimating a conditional qf with  discrete and continuous covariates suggesting a kernel based approach. 
\\
But regardless of whether parametric or non-parametric approaches are chosen, the theory concerning the validity of the correct model choice seems to keep up with the rapid development of new estimation methods only to a limited extent. To the best of our knowledge, there does not exist a testing procedure
that allows to test for quantile-specific functional (such as nonlinear) covariate effects. To fill this gap is the aim of this paper. To do so, we develop a broad approach that tackles the aforementioned challenges relevant for a wide range of applied questions simultaneously; namely (i) possibly nonlinear functional forms of covariate effects on certain conditional quantiles. (ii) possibly quantile-specific regressor effects and (iii) a more powerful semi-parametric test in the framework of many regressors, all of which offer estimation advantages compared to quantile-independent regressors, particularly in small samples. 
\\
We first suggest a general procedure for quantile regression models, where the regressors can explicitly depend on quantiles. This allows to test for the correct specification of large number of parametric models. Second, due to our general model set-up, our proposed methodology also allows to test for finite semi-parametric models.  One of such examples are B-splines for quantile regressions, where the finite number of covariates have a functional form and thus depend on the quantile \citep[][]{cardot2005quantile}. Additionally, our second test allows to test for the order and the correct number of knots of the B-spline specifications. 
The third test is developed in the framework of quantile regression models with an increasing number of covariates, known as models with many regressors. Employing the general structure, the third test can also be applied to test (semi-)parametric quantile regression models which turns out to be a more powerful testing procedure. 
The proposed valid bootstrap procedure is a practical easy-to-implement algorithm to calculate critical values of the limiting distributions.  Overall, our three tests extend
the literature on quantile regression specification tests significantly allowing better answering relevant questions in economics and further sciences;
wherever specific regressors may have a functional, nonlinear influence and/or the respective effects may vary over quantiles.
\\
The key idea of our procedure is based on the principle characterized by \citet{rothewied:2013}: We compare an unrestricted estimate of the joint distribution function of the random variable $Y$ and the vector $X$ with a restricted estimate that imposes the structure implied by the null hypothesis model. Based on a Cram\'{e}r-von Mises type measure of distances, the restricted estimate of the joint distribution can then be compared with the unrestricted one. We derive the non-pivotal limiting distribution of our test statistic and show the validity of our suggested parametric bootstrap procedure for the approximation of the critical values. To increase the power of our test, we replace the unrestricted model estimate with a quadratic B-spline. Due to the generality of our test procedure we can subsume previous specification tests for quantile regression models with i.i.d.~data as marginal cases of our procedure. Our extensive Monte Carlo (MC) simulation study in the  Supplement shows that our testing procedures are consistent and have superior power properties than existing benchmark methods, where comparisons are possible.
\\
Finally, to illustrate the power and potential of our tests, we consider two real data applications. First, the case of income inequality is treated, with a focus on differences in the conditional income quantiles between East and West Germany in a balanced panel data set. Such disparities have received considerable attention in the economic literature \citep[e.g.~][]{biewen2000income}, and also consistently played a major role in the domestic political debate. Our empirical analysis uses the German Socio-Economic Panel (SOEP) and shows that age has a predominant linear influence on income development in Germany, but for the upper $90\%$ quantile the influence of age is solely quadratic. Importantly, and in line with other studies on this topic, we find through an initial M-M  decomposition that there are still income differences between East and West Germany, which can be confirmed by our proposed testing procedure. The second application arises from energy economics. Following recent work in \citet{smith2019bayesian}, we consider spot prices from the Australian national electricity market from 2019 and analyze in which sense its conditional quantiles can be explained by different covariates. These authors have shown that the distribution is heavily skewed and far from Gaussian with complex interactions of the three covariates day of the year, the time of day and the demand. We statistically confirm that interaction effects have a substantial impact on the electricity price, especially for the lower quantiles. 
\\
The paper is organized as follows. Sec.~\ref{sec:qrt} formulates the test problem for the finite-dimensional parametric and semi-parametric model.  From this, we discuss the  many regressors model. In Sec.~\ref{asymptotics}, we provide the theoretical properties of the testing procedures and derive their limiting distributions. Sec.~\ref{sec:bootstrap} describes a practical and easy-to-implement bootstrap procedure, which provides valid coverages. In Sec.~\ref{sec:emp} we present the two empirical applications. The last Sec.~\ref{sec:conclusion} concludes.  Supplement contains all proofs of our theoretical results, as well as an extensive MC study including comparisons to existing tests and further results on the second application. 
\section{Quantile Regression Testing}\label{sec:qrt}

In this section, we introduce three specification tests for (semi-)parametric quantile regression models comparing the empirical conditional cumulative distribution function (ecdf) with the (semi-)parametric joint cdf that is based on the estimated conditional qf. We denote these tests by $S_n^{CM}$, $S_n^{CM, S}$ and  $S_n^{CM^*}$. In Sec.~\ref{sec:test1} we derive the general test principle along the lines of parametric models. In contrast to existing approaches, the test for parametric quantile regression models $S_n^{CM}$ allows the covariates $X$ to be quantile-dependent. Sec.~\ref{sec:test2} applies the general test principle to finite-dimensional semi-parametric models with the specification test denoted by $S_n^{CM, S}$. As an illustrative example, we consider B-splines,  where the degree of the spline and the dimension of the vector of knots is known and finite. In Sec.~\ref{splinesection}, we introduce a more powerful model specification testing procedure $S_n^{CM^*}$, which is illustrated on the class of parametric quantile regression models. To do so, we replace the empirical conditional cdf in the test statistic $S_n^{CM}$ with an appropriate spline representation that approximates the true joint cdf faster. The price of the higher power is that the class of true cdfs is restricted more strongly. For the approach, we need splines whose dimension grows as a function of the number of observations, i.e., the degree of the spline is fixed while the dimension of the knot vector diverges at an appropriate rate. Such models with increasing regressors are known as models with many regressors. Finally, we note that it would in principle also be possible to do this extension for $S_n^{CM,S}$ with some additional assumptions but doing so in detail is beyond the scope of this paper.

\subsection{Quantile Regression  and the General Test Principle}\label{sec:test1}

Let $Y_i\in\R$ denote the  outcome variable and $X_i\in\R^K$ the vector of explanatory variables of i.i.d.~data points for $i=1,\ldots,n$ and $K\in\mathbb{N}$.  Our aim is to test the validity of certain model specifications for quantile regression. Specifically, we consider models of the form 
\begin{align}\label{qrmodel}
    F_{Y\mid X}^{-1}(\tau\mid x)=P(x,\tau)^\top\te(\tau),
\end{align}
where $F_{Y\mid X}^{-1}(\tau\mid x)$ denotes the  qf of $Y$ conditional on $X=x\in\R^K$ at quantile $\tau$, $P(x,\tau)\in\R^{p_\tau}$ is a transformation vector of $x$ with $p_\tau\in\mathds{N}$ and $\theta(\tau)\in\R^{p_\tau}$ is the parameter vector depending on $\tau$ for all $\tau\in\mathcal{T}\subset[0,1]$. Naturally, models in which the vector of transformations does not depend on $\tau$ are captured by our approach as a special case. As noted by \cite{belloni2019} for $P(\x,\tau)\equiv P(\x)$, the above framework incorporates a variety of models such as parametric \citep[][]{koenker:2005} and semi-parametric \citep[][]{heshi} ones. However, since we allow the transformation vector $P(\x,\tau)$ to depend on the quantile $\tau$, models of the form \eqref{qrmodel} are generalizations. In parametric quantile regression models, $P(\x,\tau)$ could for instance represent a linear covariate in the lower 50\% quantile and a highly nonlinear functional regression form in the upper 50\% quantile, e.g.~$P(\x,\tau)=\x$ if $\tau\in[0,0.5]$ and $P(\x,\tau)=\sin(\x)\x^2$ for $\tau\in(0.5]$. In semi-parametric models, $P(\x,\tau)$ could represent the knot vector for cubic B-splines that differs for distinct quantiles as in our second applicationin Sec.~\ref{sec:edata}. For ease of notation, we assume $p_\tau=:p\in\mathds{N}$ for all $\tau\in\mathcal{T}$, since transformations of $\x$ that do not appear for certain quantiles $\tau$ can be  replaced by $0$. In the remainder of this subsection we assume the qf according to \eqref{qrmodel} to be specified by a parametric model, while generalizations are treated thereafter. 
\\
Our test principle is designed for the comparison of the non-parametric with the parametric joint cdf, where the latter can be expressed by means of the parametric conditional cdf. The conditional cdf $F$ of $Y$ conditioned on $X$, denoted as $F_{Y\mid X}$, in turn is induced by its  corresponding (generalized) conditional qf $F_{Y\mid X}^{-1}$ through the following equation 
\begin{align}\label{e1}
F_{Y\mid X}(y\mid x)=\int_0^1 \one_{\left \{F^{-1}_{Y\mid X}(\tau\mid x)\leq y\right\}} d\tau\qquad \forall\, y\in\R.
\end{align}
In the following, we consider the set of all conditional distribution functions satisfying \eqref{e1} given the model specification \eqref{qrmodel}, which we denote by $\mathcal{F}$, i.e.~ \begin{align}\label{scriptf}
\!\!\!\!\mathcal{F}:=\{F_{Y\mid X}(y\mid x,\te)\mid F_{Y\mid X}^{-1}(\tau\mid x)=P(x,\tau)^\top\te(\tau) \text{ for some }\theta\in\B,\ (y,x)\in\Supp\},
\end{align}
where $\Supp$ denotes the support of $(y,x)\in\R^{K+1}$ and $\B$ the class of functions $\tau\mapsto\theta(\tau)\in\Theta\subset\R^{p}$. The  specification testing problem of whether our model  \eqref{qrmodel} is correctly specified for all $\tau\in\mathcal{T}$ transfers by means of \eqref{scriptf} to hypotheses of the form 
\begin{align}\label{M1}
H_0: F_{Y\mid X}\in\mathcal{F}\quad\text{vs.}\quad H_1: F_{Y\mid X}\notin \mathcal{F}. 
\end{align}
Thus, we want to test if the conditional cdf $F_{Y\mid X}$ coincides with an element of 
 $\mathcal{F}$ from \eqref{scriptf}.
For this testing problem, we assume a unique $\te_0\in\B$ under the null hypothesis, such that $\te(\tau)=\te_0(\tau)$ for all $\tau\in\T$. This yields  $\mathcal{F}^0:=\{F_{Y\mid X}(y\mid \x, \te_0)\,\mid\, F_{Y\mid X}^{-1}(\tau\mid \x)=P(\x,\tau)^\top\te_0(\tau) \text{ for some } \te_0\in\B \forall(y,\x)\in\Supp \}$. Hence, we can reformulate \eqref{M1} as
\begin{align}\label{null}
\begin{split}
&H_0: F_{Y\mid X}(y\mid x)=F_{Y\mid X}(y\mid x, \te_0)\text{ for some } \te_0\in\B \text{ for all }(y,\x)\in\Supp\\
\text{vs. }&H_1: F_{Y\mid X}(y\mid x)\neq F_{Y\mid X}(y\mid x, \te)\text{ for all } \te\in\B \text{ for some }(y,\x)\in\Supp.
\end{split}
\end{align} 
Additionally we assume that $\te_0$ is identified under the null hypothesis through a moment condition. Specifically, let $g: \Supp \times\Theta\times \T\to \R^p$ be a uniformly integrable function whose exact form depends on $\mathcal{F}^0$, and suppose that for every $\tau \in \T$ 
\begin{align}\label{GMM}
G(\theta, \tau):=\E[g(Y,X,\theta,\tau)]=0\in\R^p
\end{align}
has a unique solution $\theta_0(\tau)$. Furthermore, under the alternative $H_1$, $\theta_0(\tau)$ remains well defined for all $\tau\in\T$ as a solution to \eqref{GMM} and can thus be thought of as a pseudo-true value of the functional parameter in this case. 
Incorporating the moment condition, we can now rewrite the null hypothesis of \eqref{M1} as
\begin{align*}
F_{Y\mid X}(y\mid \x)=F_{Y\mid X}(y\mid \x,\theta_0) \text{ for all  } (y,\x)\in\R^{K+1},
\end{align*}
with $\theta_0(\tau)$ as the unique solution to \eqref{GMM} for all $\tau\in\T$. This holds true since $\mathcal{F}^0$ is a singleton containing $F_{\cdot\mid \cdot}(\cdot \mid \cdot, \theta_0)$. Since $F_{Y\mid X}(y\mid \X)=\E[\one_{\{Y\leq y\}}\mid \X ]$, we can write the joint cdf $F$ of $Y$ and $\X$ as
\begin{align*}
F(y,\x)&=\dint_{\R^{K}} F_{Y\mid X}(y\bt \x^*)\one_{\{\x^*\leq \x\}} dF_\X(\x^*)\\
F(y,\x,\theta_0)&=\dint_{\R^{K}}  F_{Y\mid X}(y\bt \x^*, \theta_0)\one_{\{\x^*\leq \x\}} dF_\X(\x^*),
\end{align*}
where $F_\X$ denotes the marginal cdf of $\X$. From Theorem 16.10 (iii)  of \cite{billingsley:1995} it follows that the testing problem \eqref{null} can be restated as 
\begin{align}\label{null3}
\begin{aligned}
&H_0: F(y,\x)=F(y,\x,\theta_0) \text{ for all } (y,\x)\in\R^{K+1}\\
\text{vs. }&H_1: F(y,\x)\neq F(y,\x,\theta_0) \text{ for some } (y,\x)\in\R^{K+1}.
\end{aligned}
\end{align}
Further, let $S:\R^{K+1}\times\Theta\to \R$ be a function that measures the difference of the non-parametric $F(y,x)$ and the parametrized cdf $F(y,x,\theta)$ defined as  
\begin{align}\label{Sthe}
{S(y,x,\theta):=F(y,x)-F(y,x,\theta).}
\end{align}
The null hypothesis is true if $S(y,x,\theta_0)=0$ for all $(y,x)\in\Supp $, whereas $S(y,x,\te)\neq 0$ for all $\theta\neq \theta_0\in\B$ and for some $(y,x)\in\Supp $. 
The sample analog is 
\begin{align}\label{testst_emp}
S_n(y,x,\hat{\theta}_n):=\hat{F}_n(y,x)-\hat{F}_n(y,x,\hat{\theta}_n),
\end{align}
where $\hat{F}_n(y,x)$ is the empirical cdf and $\hat{F}_n(y,x,\hat{\theta})$ a parametric estimate of $F$ based on a consistent estimate $\hat{\theta}_n(\tau)$ of $\theta_0(\tau)$ for all $\tau\in\mathcal{T}$ corresponding to the underlying model assumption \eqref{qrmodel}. Under the null hypothesis, $\hat{F}_n(y,x,\hat{\te}_n)$ is a consistent estimator for $F(y,x,\te_0)$, whereas $\hat{F}_n(y,x)$ consistently estimates $F(y,x)$. In that case, $S_n(y,x,\hat{\theta}_n)$ should be close to zero for all $(y,x)\in \Supp$. If, however, the alternative holds true, then there is a vector $(y,x)\in\mathcal{S}$ for each $\theta\in\B$ such that the absolute value of the function $S_n$ from \eqref{testst_emp} is greater than zero. 
\\
To obtain an estimate for the parametrized empirical cdf $\hat{F}_n(y,x,\hat{\te}_n)$ we follow \citeauthor{cher:2013} (\citeyear{cher:2013}) and take the function $\hat{\theta}_n$ to be an approximate $Z$-estimator satisfying 
\begin{align}\label{Zestimator}
\left\lVert\hat{G}_n(\hat{\theta}_n,\tau)\right\rVert=\inf\limits_{\theta\in\Theta}\left\lVert\hat{G}_n( {\theta} ,\tau)\right\rVert+\eta_n,
\end{align} 
where the function $\hat{G}_n(\hat{\theta}_n,\tau):=n^{-1}\sum\limits_{i=1}^n g(Y_i,X_i,\theta,\tau)$ is the sample analogue of the moment condition \eqref{GMM} for every $\tau\in\T$ and for some possibly random variable $\eta_n=o_p(n^{-1/2})$. For every $\tau\in\T$ and every $(y,x)\in\Supp$, the estimator based on the testing problem \eqref{null} is\begin{align}\label{marginaly}
\begin{aligned}
&\hat{F}_{n}(y\mid x,\hat{\theta}_n)=\dint_0^{1} \one_{\{P(x,\tau)^\top \hat\te_n(\tau)\leq y\}}d\tau,\\\
&\hat\te_n(\tau)=\argmin_{\te\in\Theta}\sum\limits_{i=1}^n\left(\tau-\one_{\{y_i\leq P(x_i,\tau)^\top \te\}}\right)\left(y_i - P(x_i,\tau)^\top \te\right).
\end{aligned}
\end{align}
The integral in \eqref{marginaly} can be computed by means of standard numerical integration techniques and corresponds to the canonical quantile regression approach, i.e.~the moment function $g$ from \eqref{GMM} is given by  $g(Y,X,\theta,\tau)=(\tau-\one\{Y\leq P(X,\tau)^\top \theta(\tau)\})P(X,\tau)$ \citep[Lemma 14 of][]{cher:2013}.  Additionally, \eqref{marginaly} and other typical estimation methods fit the estimated qf $\hat{F}^{-1}_n(\tau\mid x)$ pointwise in $\tau\in\mathcal{T}$, which might induce the problem that the estimated quantile curve $\tau\mapsto\hat{F}^{-1}_n(\tau\mid x)$ violates the monotonicity constraint. This in turn may cause crossing quantile curves.
However, a violation of the monotonicity constraint does not affect the validity of the test statistic, since it is based on transformations of $\hat{F}_n(y\mid x,\hat{\te}_n)$ which is  monotone in $y$ by construction for every $x$.
Hence, a valid test statistic can be based on the differences of the non-parametric and parametric ecdfs $\hat{F}_n(y,x)$ and $\hat{F}_n(y,x,\hat{\theta}_n)$ and thus expressed as
\begin{align}\label{S_n}
\begin{aligned}
S_n&(y,x,\hat{\te}_n)=\hat{F}_n(y,x)-\hat{F}_n(y,x,\hat{\te}_n)\\
&=\frac{1}{n}\sum\limits_{i=1}^n\left(\one_{\{Y_i\leq y\}}\one_{\{X_i\leq x\}}\right)-\dint_{\R^{K}}\one_{\{x^*\leq x\}}\left(\dint_0^{1} \one_{\{P(x^*,\tau)^\top \hat\te_n(\tau)\leq y\}}d\tau\right) \ d\hat{F}_X(x^*)\\
&=\frac{1}{n}\sum\limits_{i=1}^n\left(\one_{\{Y_i\leq y\}}\one_{\{X_i\leq x\}}-\one_{\{X_i\leq x\}}\left[\dint_0^{1} \one_{\{P(X_i,\tau)^\top \hat\te_n(\tau)\leq y\}}d\tau\right]\right),
\end{aligned}
\end{align}
where the third line exploits the definition of the integral with respect to the ecdf $\hat{F}_X$. We propose a Cram\'{e}r-von Mises type (CM) test statistic $S_n^{CM}$ defined as 
\begin{align}\label{CM-test}
S_n^{CM}:=\dint \left(\sqrt{n}S_n(y,x,\hat{\te}_n)\right)^2d\hat{F}_n(y,x),
\end{align}
which is due to the quantile dependence of the covariates a generalization of existing quantile regression tests. However, if the vector of transformations $P(x,\tau)$ in \eqref{qrmodel} is independent of $\tau$ then the test statistic coincides with test statistic proposed in \cite{rothewied:2013}. It is also possible to consider a Kolmogorov-Smirnov-type test statistic
\begin{align*}
S_n^{KS}:=\sqrt{n}\sup\limits_{(y,x)\in\Supp}\left\lvert S_n(y,x,\hat{\te}_n)\right\rvert,
\end{align*}
but the CM test yields better (power) results 
\citep[][]{rothewied:2013, cher:2002}. 

\subsection{Specification Test for Semi-Parametric Quantile Regression}\label{sec:test2}

Since we introduced the general testing principle of \eqref{qrmodel} by means of the parametric model, this subsection briefly demonstrates that the general test principle is also applicable to finite dimensional semi-parametric models. This particularly addresses the fact that parametric models are often too restrictive and implausible from an applied perspective, since, amongst others, the constantly increasing complexity of data sets also makes modeling by simple functional relationships more difficult.
\\
In the following, we identify the vector of transformations $P(x,\tau)$ as basis functions  \citep[often referred to as series terms;][]{volgushev2019distributed,belloni2019,cher:2013}.
To distinguish such basis functions from the vector of transformations in the previous subsection, we use the notation $B^\cdot$ instead. Due to their widespread use, we will derive the semi-parametric test for B-splines bases, although our general test principle also allows for other semi-parametric forms such as P-splines, Fourier series or compactly supported wavelets \citep[][]{volgushev2019distributed}. For ease of well-defined expression and readability, we assume w.l.o.g.~that the vector of covariates $X\in\R^K$ is properly scaled and centered and that $k\in\mathds{N}$ uniformly spaced knots $0=t_1<\ldots<t_{k}=1$ in the interval $[0,1]$ are given. For $x=(x_1,\ldots,x_K)^\top\in[0,1]^K$ with $K\in\mathds{N}$, we identify the B-spline quantile regression model in the spirit of \eqref{qrmodel} as
\begin{align}\label{Bmodel}
F^{-1}_{Y\mid X}(\tau\mid X=x)=\sum\limits_{j=1}^K B(x_j\mid d_\tau)^\top \theta_j(\tau)
\end{align}
with ${B}(x_j\mid d_\tau):=\left({B}_0(x_j\mid d_\tau),\ldots,{B}_{M-1}(x_j\mid d_\tau)\right)^\top$ being  $M$ basis functions  of degree $d_\tau$ that are defined recursively on the vector of knots on $[0,1]$ and evaluated at $x_j$ for $j=1,\ldots, K$ \citep[cf.~][for the recursive Definion]{deboor}. For every $\tau\in\mathcal{T}$ and $j=1,\ldots,K$,  ${\theta}_j(\tau)=(\theta_  {j,0}(\tau),\ldots,\theta_{j,M-1}(\tau))^\top$ defines the corresponding functional coefficient vectors. Although both $M$ and $d_\tau$ can be conceived to depend on $j$ for $j=1,\ldots,K$ and additionally $M$ on $\tau$, we suppress these dependencies at this point due to readability and clearness. Note that for distinct quantiles $\tau$ the degree of the B-spline might differ. If $d_\tau\equiv d\in\mathds{N}$ we refer to \eqref{Bmodel} as B-spline quantile regression model of degree $d$. Since our general quantile regression model in \eqref{qrmodel} conceptually allows for multivariate covariates,  we make  \eqref{Bmodel} more flexible by adding $q^*\in\mathds{N}$ arbitrary product interaction effects 
of the form $\pi_i(x)=\prod\limits_{j\in J_i}f_j(x_j)$, where $J_i$ is an arbitrary subset of $\{1,\ldots,K\}$ for $i=1,\ldots,q^*$ and $f_j$ an arbitrary continuous function for $j\in J_i$. Thus,  \eqref{Bmodel}  generalizes to 
\begin{align}\label{genB}
F^{-1}_{Y\mid X}(\tau\mid X=x)=\sum\limits_{j=1}^q B(\pi_j(x)\mid d_\tau)^\top \theta_j(\tau) 
\end{align}
with $q=K+q^*$. For $\pi_j(x)=x_j$ and $J_i$ singletons for $j=1,\ldots,q$ with $q=K$ we receive our initial B-spline model \eqref{Bmodel}.  
The estimator for models of the form \eqref{genB} is given by 
\begin{align}\label{Bestimate}
\begin{aligned}
\hat\theta_n(\tau)=\argmin_{\theta\in\R^{q\cdot M}}&\left\{\sum\limits_{i=1}^n \rho_\tau\left(y_i-\sum\limits_{j=1}^q B(\pi_j(x)\mid d_\tau)^\top \theta_j \right) \right\}
\end{aligned}
\end{align}
where $\rho_\tau(u)=u(\tau-I(u<0))$ is the check function \citep{koenker1978regression} for $\tau\in\mathcal{T}$, $u\in\mathbb{R}$.  
In case of no misspecification, \citet{bondellnoncrossing} have shown that the unconstrained estimator \eqref{Bestimate} has the same limiting distribution as the classical constrained quantile regression estimator. Hence, in accordance with the discussion on monotonicity in Sec.~\ref{sec:test1}, $\theta_0$ can be estimated consistently based on  unconstrained  methods, noting that possible quantile curves crossing of the qf estimator does not affect the validity of the CM test statistic. Besides the well-known estimation method \eqref{Bestimate}, there are other consistent approaches.  
A prominent and easy to implement algorithm  is the \textit{divide and conquer} algorithm at fixed $\tau$. The \textit{quantile projection} algorithm, in contrast, is used to construct an estimator for the quantile process \citep[cf.~][for further details]{volgushev2019distributed2}.
\\
To develop a CM test for null hypotheses of the form \eqref{null}, we replace the estimator of the conditional qf in  \eqref{marginaly} with our estimator \eqref{Bestimate}. This yields a new conditional distribution function $\hat{F}^S_n(y\mid x,\hat{\te}_n)$. Integrating over $x$ leads to the function $S_n^S(y,x,\hat{\te}_n):=\hat{F}_n(y,x)-\hat{F}_n^S(y,x,\hat{\te}_n)$, where $\hat{F}_n^S(y,x,\hat{\te}_n)$ is the spline based estimate of the cdf in the spirit of \eqref{S_n} for fixed and finite $q$ and $M$. We then define the CM test statistic for finite-dimensional semi-parametric quantile regression models, i.e.
\begin{align}\label{CMS-test}
S_n^{CM,S}:=\dint \left(\sqrt{n}S_n^S(y,x,\hat{\te}_n)\right)^2d\hat{F}_n(y,x).
\end{align}    
The second test $S_n^{CM, S}$ is thus able to test, for instance, whether a spline is correctly specified with respect to its predefined fixed degree $d$. Consequently, questions of the form whether quadratic splines characterize a data set similarly well as cubic splines can be addressed by means of $S_n^{CM, S}$. 

\subsection{A More Powerful Testing Procedure Using Splines}\label{splinesection}

The underlying principle of the test $S_n^{CM}$ is to compare the parametric cdf induced by \eqref{null} with the non-parametric cdf. The power of this test can be improved if a spline  is used instead, i.e.~modelling the conditional qf with an appropriate spline function  used to estimate the ecdf $\hat{F}_n$. \citet{xue2010} have shown, for instance, that the estimate of the cdf with a smooth monotone polynomial spline has better finite sample properties than the empirical distributional estimate. \citet{cardot2005quantile} have generalized limiting results for quantile regression models with quantile-dependent covariates. However, the goodness and convergence rate of the spline approximation depends, in general, in a complex fashion  on the degree of the spline, the number of knots and the position of the knots which may change for increasing $n$. For a quantile regression model with quantile-independent covariates \citet{heshi} have pointed out that if the number of knots $k_n\sim (n/\log{n})^{2/5}$ and under some mild assumptions \citep[cf.~][assumptions $C1-C3$]{heshi}, the order of approximation of a quadratic monotone B-spline is $(\log{n}/n)^{2/5}$.
\\
However, in order to approximate the theoretical cdf sufficiently well, it is necessary that $M$ grows as a function of $n$, i.e.~$M$ diverges at a proper rate. This is known as quantile regression with many regressors, i.e.~a linear model with increasing dimension. Note that in this framework the true functional parameter vector $\theta_0$ also depends on $n$ \citep{belloni2019}. Consequently, for quantile regression  with many regressors,  \eqref{qrmodel} expands to
\begin{align}\label{qrmodelwithmanyreg}
    F_{Y\mid X}^{-1}(\tau\mid x)=P_n(x_n,\tau)^\top\te_{0_n}(\tau).
\end{align}
We discuss the theoretical framework of \eqref{qrmodelwithmanyreg} in Sec.~\ref{sec:asymp:semipara}.   
Let $\hat{F}_n^{S_M}$ be the spline based estimate of the cdf via the conditional qf with many regressors according to \eqref{qrmodelwithmanyreg}. Thus in many regressor models, the test statistic that is based on the difference of the parametric and semi-parametric ecdf reads
\begin{align*}
\begin{aligned}
S_n^*(y,x,\hat{\te}_n)&=\frac{1}{a_n^*}\hat{F}_n^{S_M}(y,x,\hat{\te}_n)-\frac{1}{a_n}\hat{F}_n(y,x,\hat{\te}_n),
\end{aligned}
\end{align*}
where $a_n,a_n^*$ are  scaling factors defined  in Sec.~\ref{sec:asymp:semipara}.   
This yields the new test statistic  
\begin{align}\label{CM-test_spline}
S_n^{CM^*}:=\dint \left(\sqrt{n}S_n^*(y,x,\hat{\te}_n)\right)^2d\hat{F}_n(y,x).
\end{align}
In comparison to $S_n^{CM}$, the test statistic $S_n^{CM^*}$ replaces the estimate of the ecdf $\hat{F}_n(y,x)$ in \eqref{S_n} by an appropriate spline estimate of the conditional qf via \eqref{qrmodelwithmanyreg}, which is then transformed to estimate $\hat{F}^{S_M}_n(y,x,\hat{\theta}_n)$. Note that finite-dimensional parametric models can also be tested with $S_n^{CM^*}$. In our MC simulation study, we will therefore compare the two tests $S_n^{CM}$ and $S_n^{CM^*}$, as they address questions of similar kind. It turns out that $S_n^{CM^*}$ is a more powerful testing procedure than $S_n^{CM}$, particularly in small samples.

\section{Asymptotics}\label{asymptotics}
In this section, we first derive theoretical properties of the parametric test statistic $S_n^{CM}$ in \ref{sec:asymp:para} before generalizing the statements to the semi-parametric test statistic  $S_n^{CM,S}$ and the more powerful test statistic $S_n^{CM^*}$ in Sec.~\ref{sec:asymp:semipara}.

\subsection{Theoretical Properties for Testing (Semi-)Parametric Quantile Regression Models}\label{sec:asymp:para}
In Theorem \ref{t1} below we show that the test statistic $S_n^{CM}$ has correct asymptotic size. To be able to derive large sample properties of $S_n^{CM}$, we make and discuss the following mild assumptions. Since our proposed test statistic is a generalization of existing tests, 
these assumptions modify those previously made \citep{cher:2013,rothewied:2013}. For this purpose, we restate the assumptions on compact subsets on $\mathcal{T}$. 
Let $\Theta$ be an arbitrary subset of $\R^p$ and $\T:=[\varepsilon,1-\varepsilon]$ with $\varepsilon\in(0,0.5)$. 
\begin{assumption}\label{ass1}\text{} 
\begin{enumerate}[i.)]
\item $P(X,\tau)$ is $L_2$-bounded in $[0,1]$ and continuous in $X$.\label{A11}
\item Let  $\bigcup\limits_{l=1}^L I_l=\T$, $L\in\IN$, $I_l$ compact for $l=1,\ldots,L$ and $I_{l_1}\cap I_{l_2}$ a singleton for $l_1\neq l_2$.\label{A12}
\item For each $\tau\in I_l$ with $l=1,\ldots,L$, $G(\cdot, \tau):\Theta\to \R^p$ possesses a unique zero at $\te_0\in interior(\Theta)$ such that $G(\theta_0,\tau)=0$ for all $\tau\in\mathcal{T}$ and for some $\delta>0$, $\mathcal{B}:=\bigcup\limits_{\tau\in\I_l}B_\delta(\theta_0)$ is a compact subset of $\R^p$ contained in $\Theta$ for $l=1,\ldots,L$.\label{A13}
\item Further,  $G(\cdot,\tau)$ has an inverse $G^{-1}(x,\tau):=\{\theta\in\Theta\mid G(\theta,\tau)=x\}$ that is continuous at $x=0$ uniformly in $\tau\in I_l$ for all $l=1,\ldots,L$ with respect to the Hausdorff distance.\label{A14}
\item The mapping $(\te,\tau)\mapsto g(\cdot, \te,\tau)$ is continuous at each $(\te(\tau),\tau)\in\Theta\times I_l$ for all $l=1,\ldots,L$  with probability one and $(\te,\tau)\mapsto G(\te,\tau)$ is continuously differentiable at $(\te_0(\tau),\tau)$ with uniformly bounded derivative on $\mathcal{T}$.\label{A15}
\item The function $\dot{G}(\te,\tau):=\partial_\theta G(\theta,\tau)$ is non-singular at $\theta_0(\cdot)$ uniformly over $\tau\in I_l$ with $l=1,\ldots,L$.\label{A16}
\item The function set $\mathcal{G}_l = \{g( Y,X,\te,\tau)\mid (\te,\tau)\in \Theta\times I_l)\}$ is $F_{YX}$-Donsker for all $l=1,\ldots,L$ with a square integrable envelope $\tilde{G}$ for $\bigcup\limits_{l=1}^L \mathcal{G}_l.$ \label{A17}
\item The mapping $\theta\mapsto F(\cdot\mid \cdot,\theta)$ is Hadamard differentiable for all $\theta\in\mathcal{B}(\mathcal{T},\Theta)$ with derivative $h\mapsto\dot{F}(\cdot\mid \cdot,\theta)[h]$ \label{A18}
\end{enumerate}
\end{assumption}
\noindent Due to quantile dependence of the regressors $X$, we further require continuity of the function $P(X,\tau)$ in $X$, which is provided by
Assumption \ref{ass1}\ref{A11}. Assumption \ref{ass1}\ref{A12} ensures that there is a finite and compact decomposition of the unit interval. This is required since we consider Donsker classes in the proof of Theorem~\ref{t1}. We are using the fact that the union of Donsker classes is also Donsker \citep[see][ Section~3.8]{dudley_2014}. Assumptions \ref{ass1}\ref{A12}--\ref{A17} guarantee the regularity of our estimator $\hat{\theta}_n $ and ensure that a functional central limit theorem can be applied to $Z$-estimator processes (see Corollary \ref{coro} in  Supplement \ref{sec::A1}). Assumption \ref{ass1} \ref{A18} is a smoothness condition. Together with the functional delta method it implies that the restricted cdf estimator process 
\begin{align}\label{con1}
(y,x)\mapsto \sqrt{n}\left(\hat{F}_n(y,x,\hat{\theta}_n)-F(y,x,\theta)\right)
\end{align} 
is $F_{YX}$-Donsker. This convergence can be shown to be jointly with that of the ecdf process $(y,x)\mapsto \sqrt{n}\left(\hat{F}_n(y,x)-F(y,x)\right)$ to a Brownian bridge by some standard arguments given in Lemma \ref{l2} in  Supplement \ref{sec::A1}. Applying the continuous mapping theorem yields the following proposition. 
\begin{theorem}\label{t1}
If Assumption \ref{ass1} 
is satisfied, then the following statements hold: 
\begin{enumerate}[i.)]
\item\label{t1i}Under the null hypothesis $H_0$ in \eqref{null3},
\begin{align*}
S_n^{CM}\stackrel{d}{\rightarrow}\dint\left(\G_1(y,x)-\G_2(y,x)\right)^2dF_{YX}(y,x),
\end{align*}
where $(\G_1,\G_2)$ is a bivariate zero mean Gaussian processes with
\begin{align*}
&\G_2(y,x):=\int\G_2^+(y,x^*)\one_{\{x^*\leq x\}}dF_X(x^*)+\int F(y\,\mid \, x^*)\one_{\{x^*\leq x\}}d\G_1(\infty,x^*),
\end{align*}
where $\G_2^+(y,x)$ is the limiting Gaussian process of $\sqrt{n}\left(\hat{F}_n(y\,\mid \, x, \hat{\te}_n)-F(y\,\mid \,x,\te)\right)\in\ell^\infty(\mathcal{\Supp})$ defined in Lemma \ref{l2}. Moreover,
\begin{align*}
Cov(\G_1(y_1,x_1),\G_2(y_2,x_2))=\lim_{n \rightarrow \infty} n\ Cov\left(\hat{F}_n(y_1,x_1)-F(y_1,x_1),\hat{F}_n(y_2,x_2,\hat{\theta}_n)-F(y_2,x_2,\theta)\right).
\end{align*}
\item Under any fixed alternative, i.e., when the data are distributed according to some $F$ that satisfies the alternative hypothesis ${H}_1$ in \eqref{null3},
\begin{align*}
\lim\limits_{n\to\infty}P(S_n^{CM}>\varepsilon)=1 \text{ for all constants }\varepsilon>0.
\end{align*}
\end{enumerate}
\end{theorem} 
\noindent  Theorem \ref{t1} ensures distributional convergence of the test statistic $S_n^{CM}$ and further shows that the non-parametric ecdf $\hat{F}_n(y,x)$ and the parametric ecdf $\hat{F}_n(y,x,\hat{\theta}_n)$ differ with probability one under the alternative hypothesis. Hence in case of misspecification, the power of the test statistic $S_n^{CM}$ converges to one as $n$ approaches infinity. 
Based on the generality of Assumption \ref{ass1} and the proof structure in  Supplement \ref{pro1}, the statements from Theorem \ref{t1} can also be extended to semi-parametric quantile regressions models with fixed $q$, $M$ as discussed in Sec.~\ref{sec:test2}. Thus, we have
\begin{corollary}\label{c2}
If Assumption \ref{ass1} is satisfied, then the following statements hold:
\begin{enumerate}[i.)]
\item Under the null hypothesis $H_0$ in \eqref{null3},
\begin{align*}
S_n^{CM,S}\stackrel{d}{\to}\int \left(\G_1(y,x)-\G^S_2(y,x)\right)^2dF_{YX}(y,x),
\end{align*}
where $(\G_1,\G_2^S)$ is a bivariate zero mean Gaussian processes with
\begin{align*}
&\G_2^S(y,x):=\int\G_2^{S^+}(y,x^*)\one_{\{x^*\leq x\}}dF_X(x^*)+\int F(y\,\mid \, x^*)\one_{\{x^*\leq x\}}d\G_1(\infty,x^*),
\end{align*}
where $\G_2^{S^+}(y,x)$ is the limiting Gaussian process of $\sqrt{n}\left(\hat{F}_n^S(y\,\mid \, x, \hat{\te}_n)-F(y\,\mid \,x,\te)\right)\in\ell^\infty(\mathcal{\Supp})$.  Moreover,
\begin{align*}
Cov(\G_1(y_1,x_1),&\G_2^{S}(y_2,x_2))=\\
&\lim_{n \rightarrow \infty} n\ Cov\left(\hat{F}_n(y_1,x_1)-F(y_1,x_1),\hat{F}_n^S(y_2,x_2,\hat{\theta}_n)-F(y_2,x_2,\theta)\right).
\end{align*} 
\item\label{c2ii} Under any fixed alternative, i.e., when the data are distributed according to some cdf $F$ that satisfies the alternative hypothesis ${H}_1$ in \eqref{null3},
\begin{align*}
\lim\limits_{n\to\infty}P(S_n^{CM,S}>\varepsilon)= 1 \text{ for all constants }\varepsilon>0.
\end{align*}
\end{enumerate}
\end{corollary}
\noindent  In empirical applications, however, it is often common to estimate series terms with smoothing penalty parameters $\lambda_j$ for $j=1,\ldots,p$, since this avoids overfitting the data. Imposing the assumption that the penalty parameters are $\lambda_j=o(n^{1/2})$ for $j=1,\ldots,p$ the penalties can be asymptotically ignored. This indicates that Corollary \ref{c2} is also valid in case of penalized quantile regression \citep{lian2015simultaneous}.

\subsection{Theoretical Properties for the More Powerful Test}\label{sec:asymp:semipara}

In the context of quantile regression with many regressors, i.e.~the number of regressors diverges at a proper rate, we need to introduce some additional notation: Since the dimension $K$ and  the true distribution $F$ of  i.i.d.~samples $(X_i,Y_i)\in\R^{K+1}$ for $i=1,\ldots,n$ can depend on $n$, we consider triangular arrays. For brevity of notation, we omit the index $n$ in the following and we write $P_i=P(X_i,\tau)$ and $P=P(X,\tau)$. Let $\lambda_{min}(A)$ and $\lambda_{max}(A)$ be the smallest and largest eigenvalue of a matrix $A$. By $\rVert b\lVert$ we denote the $L^2$-norm of a vector $b$. Moreover, we set $a_n:={\sqrt{n}}/{\rVert P(\xo)\lVert}$ and $a_n^*:={\sqrt{n}}/{\rVert B(\xo)\lVert}$, respectively.
\\
Imposing the assumptions from \citet{volgushev2019distributed} adapted to quantile regression with quantile-dependent series terms enables us to replace the qf by an appropriate (spline) estimator and thus to derive large sample properties for our third test statistic $S_n^{CM^*}$: 
\begin{assumption}\label{ass2}\text{} 
\begin{enumerate}[i.)]
\item For $p:=qM$, assume that $\left\lVert P_i\right\rVert\leq \xi_p = O(n^a)$ almost surely with $a>0$, and that $\frac{1}{M^*}\leq \lambda_{min}(\E[PP^\top])\leq \lambda_{max}(\E[PP^\top])\leq M^* $ holds uniformly in $n$ and $\tau\in\mathcal{T}$ for some fixed $M^*>0$.\label{A21}
\item The conditional distribution $F_{Y\mid X}(y\mid x)$ is twice differentiable w.r.t.~$y$.
We denote the corresponding derivatives by $f_{Y\mid X}(y\mid x)$ and $f^\top_{Y\mid X}(y\mid x)$. Assume
that $\bar{f}:= \left\lvert\sup_{y,x}f_{Y\mid X}(y\mid x)\right\rvert<\infty$ and $\bar{f^\top}:= \sup_{y,x}\left\lvert f^\top_{Y\mid X}(y\mid x)\right\rvert<\infty$ uniformly in $n$.\label{A22}
\item Assume there exists a constant $f_{min} >$ 0 such that \label{A23} $\inf\limits_{\tau\in\mathcal{T}}\inf\limits_x f_{Y\mid X}(F^{-1}_{Y\mid X}(\tau\mid x)\mid x)\geq f_{min}$.
\item For each $x$, the basis vector $P$ has zeroes in all but at most $r$ consecutive entries, where $r$ is fixed. Moreover,
$\sup_{\tau,x}\E[\mid P^\top \tilde{J}_m(\tau)^{-1}P\mid ]=O(1)$, where $\tilde{J}_m(\tau):=\E[PP^\top f_{Y\mid X}(F^{-1}_{Y\mid X}(\tau\mid x)\mid X) ]$.\label{A24}
\item Assume that $\xi_p^4(\log n)^6=o(n)$ and letting $c_n:=\sup_{\tau,x}\left\lvert F^{-1}_{Y\mid X}(\tau\mid X)-P^\top \hat{\te}_n(\tau)\right\rvert$ with $c_n^2=o(n^{-1/2})$\label{A25}.
\end{enumerate}
\end{assumption} 
\noindent  As mentioned in \citet{volgushev2019distributed}, Assumption \ref{ass2} \ref{A21} claims rescaling in case of B-splines and for linear models with increasing dimension $P(X,\tau)$ to be bounded for all $\tau$. Assumptions \ref{ass2} \ref{A22}.)--\ref{A23}.) are fairly standard. Assumptions \ref{ass2} \ref{A24}.) and \ref{A25}.) imply that for any sequence satisfying $c_n=o(1)$ and that the smallest eigenvalues of the matrix $J_m(\tau)$ are bounded away from zero uniformly in $\tau$ for all $n$. 
Using Theorem 2.4 of \citet{volgushev2019distributed} showing that a standardized version of the quantile series terms process converges to a centered Gaussian process we have
\begin{theorem}\label{t2}
If Assumptions \ref{ass1} and \ref{ass2} are satisfied, then the following statements hold:
\begin{enumerate}[i.)]
\item Under the null hypothesis $H_0$ in \eqref{null3},
\begin{align*}
S_n^{CM^*}\stackrel{d}{\to}\int \left(\G_2(y,x)-\G^{S_M}_2(y,x)\right)^2dF_{YX}(y,x),
\end{align*}
where $(\G_2,\G_2^{S_M})$ are Gaussian processes with zero mean given in  Supplement \ref{sec::A2}. 
\item\label{t2ii} Under any fixed alternative, i.e., when the data are distributed according to some cdf $F$ that satisfies the alternative hypothesis ${H}_1$ in \eqref{null3},
\begin{align*}
\lim\limits_{n\to\infty}P(S_n^{CM^*}>\varepsilon)= 1 \text{ for all constants }\varepsilon>0.
\end{align*}
\end{enumerate}
\end{theorem}
\noindent Theorem \ref{t2} ensures that the test $S_n^{CM^*}$ is asymptotically normal and has power in case of misspecification. The convergence statements from Corollary \ref{c2} and Theorem \ref{t2} hold for additive univariate series terms, including, for instance, univariate B-splines with product interacting covariates. In line with \citet{volgushev2019distributed}, we further conjecture that such arguments as those given in the proofs (cf.~Supplement \ref{sec::A1}) can also be applied to multivariate splines and thus in particular to tensor product B-splines considered later in Sec.~\ref{sec:edata}.  Therefore, convergence statements from  Corollary \ref{c2} and Theorem \ref{t2} can be extended to a more general class of (multivariate) splines. Inspired by this observation and our second application, we show empirically that the test statistic $S_n^{CM^*}$ based on tensor product B-splines also yields a reasonable sized testing procedure with large power (cf.~Tables \ref{3d21} and \ref{3d3}, Supplement \ref{secMC}). However, a detailed theoretical investigation of this interesting topic is beyond the scope of this paper and  left for future research.

\section{Bootstrap}\label{sec:bootstrap}
To obtain critical values for our test $S_n^{CM}$, we therefore propose a semi-parametric bootstrap procedure. This procedure is reasonable from a practical point of view, since it avoids to estimate the null distribution directly, including a complex covariance structure. 

\subsection{Semi-Parametric Bootstrap Procedure}\label{sec:bootstrap_procedure}
The idea of our semi-parametric bootstrap is to generate synthetic data that is consistent with the assumptions under the null hypothesis. Since the qf is already known according to our null hypothesis, our bootstrap procedure is based on the principle of inverse sampling transformation, which provides a method to generate samples from arbitrary distributions.  Thus, the bootstrap mimics the distribution of the data under the null hypothesis, even though the data might be generated by an alternative distribution. The procedure works as follows. Let $B$ be the number of bootstrap samples. Then 
\begin{enumerate}[i.)]
\item Draw $B$ independent bootstrap samples of covariates $\{X_{b,i}\mid 1\leq i\leq n \}_{b=1,\ldots, B}$ of size $n$ with replacement from $\{X_i\mid  1\leq i\leq n  \}$.
\item For every $i=1,\ldots,n$ put ${Y}_{b,i}=\hat{F}_n^{-1}(U_{b,i}\mid  {X_{b,i}}, \hat{\theta}_n)$, where $\{U_{b,i}\mid  1\leq i\leq n \}$ is a simulated i.i.d. sequence of standard uniformly distributed random variables.
\item Use the bootstrap data $\{(Y_{b,i},X_{b,i})\mid  1\leq i \leq n\}_{b=1,\ldots,B}$ to calculate $B$ bootstrap versions of the test statistic $S_n^{CM}$ from \eqref{CM-test}, i.e.~for $b=1,\ldots,B$ compute 
\begin{align*}
S_{n,b}^{CM}:=\dint \left(\sqrt{n}S_{n,b}(y_b,x_b,\hat{\te}_n)\right)^2d\hat{F}_n(y_b,x_b).
\end{align*}
\item For $q\in(0,1)$, determine the critical value $\hat{c}_n(q)$ such that
\begin{align*}
\frac{1}{B}\sum\limits_{b=1}^B\one_{\{S_{n,b}^{CM}>\hat{c}_n(q)\}}\stackrel{}{=}q.
\end{align*}
\end{enumerate}
With the bootstrap procedure described above, we can calculate critical values $\hat{c}_n(q)$ for \eqref{CM-test}. Critical values for \eqref{CMS-test} and \eqref{CM-test_spline} can be obtained in the same manner if the test statistic $S_{n,B}^{CM}$  is replaced by its counterparts, i.e.~$S_{n,B}^{CM,S}$ or  $S_{n,B}^{CM^*}$.

\subsection{Validity of the Bootstrap Procedure}
Finally, according to \citet{rothewied:2013}, we show that the proposed bootstrap procedure computes the correct critical value for our test statistic \eqref{CM-test}. 
This does not require any further assumptions. Assumption \ref{ass1} ensures that the bootstrap consistently estimates the limiting distribution for \eqref{CM-test}. Under the null hypothesis and  any fixed alternative \eqref{null}, the bootstrap critical values can be shown to be bounded in probability. Thus,
\begin{theorem}\label{T3}
Under Assumption \ref{ass1}, the following statements hold true for every $\alpha\in(0,1)$:
\begin{enumerate}[i.)]
\item Under the null hypothesis $H_0$ in \eqref{null3}, we have that 
\begin{align*}
\lim\limits_{n\to\infty}P(S_n^{CM}>\hat{c}_n(\alpha))=\alpha
\end{align*}
\item Under any fixed alternative $H_1$ in \eqref{null3}, we have that 
\begin{align*}
\lim\limits_{n\to\infty}P(S_n^{CM}>\hat{c}_n(\alpha))=1.
\end{align*}
\end{enumerate}
\end{theorem}
\noindent In order to study the behavior of the Cram\'{e}r-von Mises type test statistics $S_n^{CM}, S_n^{CM,S}$ and $S_n^{CM^*}$ in finite samples, we conducted an extensive MC study, whose results are reported in  Supplement \ref{secMC}. Overall, the MC study has shown that our proposed testing procedures are also consistent based on critical values  obtained via the bootstrap procedure described in Sec.~\ref{sec:bootstrap_procedure} and have superior power properties compared with three benchmark tests (cf.~Supplement \ref{sub41}), even in small samples. The testing procedures works for both, univariate and multivariate DGPs (including product interacting or more complex tensor product covariates) and can also test models with quantile-dependent regressors. Even weakly misspecified models are detected in sufficiently large sample sizes.

\section{Empirical Illustrations}\label{sec:emp}
\subsection{Income Disparities Between East and West Germany}\label{subsec51}
In this section, we apply the bootstrap version of the specification test $S_n^{CM^*}$ to conditional income distributions in Germany. We utilize information from the German Socio-Economic Panel \citep[SOEP,][]{wagnerfrickschupp}. More specifically, we consider real gross annual
personal labor income in Germany as defined in \citet{bachsteiner} from 2001 to 2010 as our response $Y$. We deflate the incomes by the consumer price index \citep{StatistischesBundesamt2012}, setting 2010 as our base year. Thus, all
incomes are expressed in real-valued 2010 Euros from here on.
Following the standard literature, we focus on  incomes of males in full-time
employment \citep[see, among others,][]{dustmann2009,card2013} in the age range 20--60. This yielded $7220$ individuals and is the data set that was also used in \citet{klein2015bayesian}. The variables $age$, $origin$ (dummy for East/West Germany) and $year$ are available as covariates, see Table \ref{descr} for a full description of the data. 
\\
To obtain an estimate of the qf, 
we first regress  income on the dummy coded variable $year$ and then performed a linear quantile regression using the variables $age$ or $age^2$ on the residuals. We consider this approach justified since four out of six tests did not reject the null hypothesis that there is no correlation between $age$ and $year$ dummies and $age^2$ and $year$ dummies, respectively. This approach takes into account that income increases at the beginning of employment, peaks in middle age and finally decreases \citep{creedy,luong, klein2015bayesian} . We next conduct a M-M decomposition \citep{machadomata, landmesser}, of the $year$-adjusted dataset conditional on  $origin$. 
\begin{table}[htbp]
\spacingset{1.6}\footnotesize
\centering
\caption{\footnotesize{Description of the German labor income data from 2001 to 2010}}
\def\arraystretch{0.85}
\begin{tabular}{p{1.25cm}p{4.75cm}p{4cm} p{2cm}}
\hline\hline\\[-4.5ex]
\text{ } & \multicolumn{3}{l}{\text{Description}}\\ 
 \cmidrule(lr){1-1}\cmidrule(lr){2-3}\\[-4.5ex]
\text{$Y$} &\text{gross market labor income (in \euro)}, (\text{continuous $1257$ $\ \leq Y\leq 280092$, $\ \text{average}=46641$})&\\[+.5ex]
  \text{$origin$} &\text{indicator for East or West (binary, -1=West (73.8\%), 1=East (26.2\%))}&\\[+.5ex]
    \text{$age$} &\text{age of the male in year (continuous, $20\leq age \leq 60$, average = 38)}&\\[+.5ex]
      \text{$years$} &\text{time in years (categorical, $2001\leq year \leq 2010$, 10 years)}&\\[+.5ex]
\hline
\\[-4.5ex]
\text{Sample} & \text{Description}&average (std.) income  & observations\\ 
 \cmidrule(lr){1-1}\cmidrule(lr){2-2}  \cmidrule(lr){3-3} \cmidrule(lr){4-4} \\[-4.5ex] 
$Ger$ &complete sample& $51026$\euro\ \ ($30569$\euro) &$n=7220$ \\[+.5ex]
$West$ &sub-sample ($origin= -1$)& $55141$\euro\ \ ($31494$\euro) &$n=5325$ \\[+.5ex]
$ East$&sub-sample ($origin= 1$)&$39463$\euro\ \ ($24336$\euro) &$n=1895$ \\[+.5ex]
\hline\hline
\end{tabular}
\caption*{Incomes. The table summarizes the descriptive statistics of the German labor income data.}
\label{descr}\vspace{-.75cm}
\end{table}
For the decomposition we assume that the qf of the income $Y$ can be represented as a function of the form $F_{Y\mid X}^{-1}(\tau\mid X)=P(X,\tau)^\top\theta (\tau)$ with $X$ consisting of the variables $age$ or $age^2$. Specifically, we consider here three different linear quantile regression models: The first model describes an entirely linear effect of the regressor $age$ on income for all quantiles $\tau\in(0,1)$, i.e.~$P(X,\tau)=age$ for all $\tau\in (0,1)$. The second models a quadratic influence of age on income for all quantiles $\tau\in(0,1)$, i.e.~$P(X,\tau)=age^2$ for all $\tau\in (0,1)$.  And finally, the third model considers the sum of the regressors $age$ and $age^2$ that are constant for all quantiles $\tau\in(0,1)$, i.e.~$P(X,\tau)=age+age^2$ for all $\tau\in (0,1)$.  
Due to the probability integral transform theorem the sequence 
$P(X,\tau_i)^\top \hat{\theta}_n(\tau_i)$ for $\tau_i\stackrel{i.i.d.}{\sim}U(0,1)$, $i=1,\ldots,n$ constitutes a random sample from the estimated conditional distribution of  $Y$ given the covariates $X$ \citep{machadomata}. In order to obtain the difference between East and West, first, the coefficients for East  ($\hat{\theta}_E(\tau)$) and  West ($\hat{\theta}_W(\tau)$) for $\tau\in\{0.1,0.2,\ldots,0.9\}$ are estimated on the basis of the disjoint subsets of the covariates for East ($X_E$) and West ($X_W$) and the corresponding income in the East ($Y_E$) and West ($Y_W$). Second, we draw  $B\in\mathds{N}$ random samples  $X_E^i$ and $X_W^i$ for $i=1,\ldots,B$ with replacement from the corresponding covariate subsets $X_E$ and $X_W$, respectively to obtain a random sample via the probability integral transform for the  distribution of the income $Y^i_l$, $i=1,\ldots,B$, $l=E,W$. Thus, the estimated income difference $\hat{\Delta}_y$ for incomes in the East $Y_E$/West $Y_W$ can  be decomposed according to M-M into 
\begin{align}\nonumber
\hat{\Delta}_Y&={F}^{-1}_{Y_E\mid X_E}(\tau\mid X_E)-{F}^{-1}_{Y_W\mid X_W}(\tau\mid X_W)\\
&=\frac{1}{B}\sum\limits_{b=1}^B\left(\left(P(X_E^b,\tau)-P(X_W^b,\tau)\right)\hat{\theta}_E(\tau)+\left(\hat{\theta}_E(\tau)-\hat{\theta}_W(\tau)\right)P(X_W^b,\tau)\right),\label{summands_MM}
\end{align}
where the first summand of \eqref{summands_MM} is the explained, while the second summand depicts the unexplained difference. 

\begin{table}[htb]
\spacingset{1.6}\footnotesize
\begin{center}
\caption{\footnotesize{Decomposition of the West/East income differential}}
\bgroup
\def\arraystretch{0.725}
\begin{tabular}{ p{2.2cm}rrrrrrrrr}
 \hline
\hline\\[-3.ex]
\textit{quantile} $\tau$&  $ 0.1$ &$ 0.2$ &$ 0.3$ &$ 0.4$ &$ 0.5$ &$ 0.6$ &$ 0.7$ &$ 0.8$&$ 0.9$   \\ 
\hline\\[-3ex]
$\textit{raw gap}$& 
-35.49 &-32.4& -33.28& -29.06& -28.38& -26.44& -26.21& -26.93& -27.38\\
& \multicolumn{9} {c}{$\textit{age}$ }\\
\cmidrule{2-10}
$\textit{M-M gap}$		&-39.87	&-36.64	& -33.52& -31.62& -31.45&-29.37&-29.68& -28.83& -25.89\\
$\textit{Explained}$		&-1.84	&-3.65	&  -2.06& -2.99 & -2.35 &-2.19 &-0.85 &-0.31& -0.1\\
$\textit{Unexplained }$	&-38.02 &-32.98	& -31.47& -28.63& -29.1 &-27.18&-28.83& -28.52& -25.8\\
$\textit{\%Explained}$			&4.63 	&9.97	&   6.13& 9.45  &7.46   &7.46  &2.87  &1.07& 0.38\\
$\textit{\%Unexplained}$		&95.37 	&90.03 	&  93.87& 90.55 &92.54  &92.54 & 97.13& 98.93& 99.62\\
$\textit{Residuals}$				&4.37 	&4.23 	&   $\mathbf{0.24}$& 2.56  &3.07   &2.93  &3.46  &1.9& $\mathbf{-1.49}$\\
& \multicolumn{9} {c}{$\textit{age}^2$ }\\
\cmidrule{2-10}
$\textit{M-M gap}$		 &-36.41& -38.13& -37.52& -35.49& -31.21& -31.96& -32.26& -31.55& -34.72\\
$\textit{Explained}$	 &-3.07 &  -6.08&  -4.49&  -6.43&  -2.81&   -2.6&  -3.92&  -4.35&  -8.82\\
$\textit{Unexplained }$	 &-33.35& -32.05& -33.04& -29.06&  -28.4& -29.36& -28.34&  -27.2& -25.89\\
$\textit{\%Explained}$	 &8.42  &  15.94&  11.96&  18.12&   8.99&   8.13&  12.15&   13.8&  25.41\\
$\textit{\%Unexplained}$&91.58 &  84.06&  88.04&  81.88&  91.01&  91.87&  87.85&   86.2&  74.59\\
$\textit{Residuals}$	 &$\mathbf{0.92}$  &   5.73&   $\mathbf{4.24}$&   6.44&   2.83&   5.51&   6.05&   4.62&   7.33\\
& \multicolumn{9} {c}{$\textit{age}$+$\textit{age}^2$  }\\
\cmidrule{2-10}
$\textit{M-M gap}$		 &-33.39& -31.80& -33.16& -30.28& -28.49& -28.90& -27.61& -28.25& -25.69\\
$\textit{Explained }$    &  2.03&   1.55&  -1.44&   -1.5&   0.13&   0.31&   0.33&  -3.09&   1.67\\
$\textit{Unexplained}$   &-35.42& -33.35& -31.72& -28.78& -28.62& -29.21& -27.94& -25.16& -27.36\\
$\textit{\%Explained}$  &  6.09&   4.89&   4.34&   4.94&   0.45&   1.09&   1.19&  10.95&   6.49\\
$\textit{\%Unexplained}$& 93.91&  95.11&  95.66&  95.06&  99.55&  98.91&  98.81&  89.05&  93.51\\
$\textit{Residuals}$     &$\mathbf{ -2.10}$ &$\mathbf{  -0.61}$ &$\mathbf{  -0.12}$ &$\mathbf{   1.22}$ &$\mathbf{   0.11}$ &$\mathbf{   2.45}$ &$\mathbf{   1.40}$ &$\mathbf{   1.32}$ &$\mathbf{  -1.69}$ \\[.5ex]
\hline
 \hline
\end{tabular}
\egroup
\caption*{\spacingset{1.0}\footnotesize Incomes. The covariates used for the quantile
regressions are $age$ (rows 4-9), $age^2$ (rows 11-16) and the sum of these two variables (rows
18-23). The second row $raw\ gap$ depicts the observed income gap between East and West. Remaining rows show three different M-M decompositions using $age$, $age^2$ and $age+age^2$ as covariates for the quantile regression models. The rows $\textit{M-M gap}$ are the estimated gap of the income difference. The quantiles $\tau$ range from $0.1$ to $0.9$. The number of bootstrap replications is equal to $2500$. All numbers are in percent. Totals may not sum exactly to 100\% due to rounding.}
\label{highloweduc}\vspace{-.8cm}
\end{center}
\end{table}
\noindent Table \ref{highloweduc} summarizes results from the counterfactual analysis described above.  The covariates used for the quantile regressions are $age$ (rows 4--9), $age^2$ (rows 11--16) and the sum of these two variables (rows 18--23). The results suggest that there is a significant income gap between East and West Germany over the period considered, which is particularly striking in the first line, where the observed income differences ranges from $26.21\%$ to $35.49\%$. However, the income difference between the smallest quantile $\tau=0.1$ and the largest $\tau=0.9$ decreases by about eight percent. It cannot be assumed that the model is sufficiently well specified by a single covariate $age$ or $age^2$ for all quantiles due to high residuals ($4.37$ for $\tau=0.1$ and $7.33$ for $\tau=0.9$), indicating misspecification. However, the covariate $age^2$ seems to be appropriate for the smallest quantile $0.1$ (residual of $0.92$ in Table \ref{highloweduc}), while a linear effect of age to income seems to prevail in higher quantiles ($-1.49$ in Table \ref{highloweduc}).  In contrast, the additive model $age+age^2$ seems to capture the income effect for all quantiles quite well due to moderate residuals (cf.~last row \textit{Residuals} in Table \ref{highloweduc}). For all decompositions it holds, that $age$ and $age^2$ contribute a maximum of $16\%$ to the explanation of the income difference between East and West Germany (except highest quantile in $age^2$, i.e.~$25.41$). Due to the different residuals and the different explanatory power of the income gap between East and West for the quantile regressions based on $age$ or $age^2$, it seems reasonable to assume that $age$ and $age^2$ have different effects for different quantiles. For example, the residual of the $30\%$ quantile of $age$ ($0.24$ in Table \ref{highloweduc}) is about $18$ times smaller than the residual of the corresponding quantile regression using $age^2$ as explanatory variable ($4.24$ in Table \ref{highloweduc}) . It is therefore reasonable that the a linear effect of age dominates in the $\tau=0.3$ quantile. The emerging, more general question is at which quantiles $age$ has a linear or quadratic effect on incomes. This can be answered with the help of our proposed test $S_n^{CM^*}$.

\begin{minipage}[]{0.5\textwidth}
\begin{align*}
\begin{aligned}
\text{S1}:& \ F_{Y\mid X}^{-1}(\tau\mid x)=\begin{cases}
x^\top\te_0\ , \text{ if } 0.1\leq\tau\leq 0.9\\
(x^2)^\top\te_0, \text{ otherwise}
\end{cases}\\
\text{S2}:& \ F_{Y\mid X}^{-1}(\tau\mid x)=\begin{cases}
(x^2)^\top\te_0, \text{ if } 0\leq\tau\leq 0.1\\
x^{\top}\te_0\ , \text{ otherwise}
\end{cases}\\
\end{aligned}
\end{align*}
\end{minipage}
\hfill
\begin{minipage}[]{0.5\textwidth}
\begin{align*}
\begin{aligned}
\text{S3}:& \ F_{Y\mid X}^{-1}(\tau\mid x)=\begin{cases}
x^{\top}\te_0\ , \text{ if } 0\leq\tau\leq 0.9\\
(x^2)^\top\te_0, \text{ otherwise}
\end{cases}\\
\text{S4}:& \ F_{Y\mid X}^{-1}(\tau\mid x)=x^\top\te_0\\
\text{S5}:& \ F_{Y\mid X}^{-1}(\tau\mid x)=(x^2)^\top\te_0
\end{aligned}
\end{align*}
\end{minipage}
\\
For this purpose, we have defined five different model specifications  S1--S5, which should take into account the observations made in Table \ref{highloweduc}.  Specifications S1--S3 describe quadratic dependencies in the upper or lower quantiles. Specification S4 and S5 model a completely linear and quadratic dependence structure in the covariate, respectively.   
\\
The testing procedure is applied to the sub-samples East and West as well as to the complete data set.  We estimate the function $\hat{F}_n^S(y,x,\theta)$ in \eqref{CM-test_spline} by a cubic spline with second order difference penalty, setting the basis dimension to $20$ using the R package \texttt{qgam}.  The smoothing parameter $\lambda$  is estimated using the restricted maximum likelihood (REML) procedure within the package. We then re-estimate the models with  optimized smoothing parameter  and compute our test statistic $S_n^{CM^*}$  for $\tau\in\{0.1, 0.2,\ldots,0.9\}$. Since the sample sizes for East, West and All differ and in order to make the results comparable, we computed the rejection rates of sub-samples of East, West and All of size $n=500, 1000, 1500$. The reason for considering different samples is, similarly to \citet{rothewied:2013}, that consistent specification tests detect also small deviations from the null hypothesis in large samples, so that smaller samples are more appropriate for model comparisons. We repeated this procedure for every sub-sample a total of $501$ times and refer to it as \textit{sub-samplings} in the following.   Table \ref{metatest} summarizes the resulting rejection rates of the test statistic $S_n^{CM^*}$. 
\begin{table}[htbp]
\spacingset{1.6}\footnotesize
\centering
\caption{\footnotesize{Rejection frequencies of the test statistic $S_n^{CM^*}$}}
\footnotesize
\def\arraystretch{0.85}
\begin{tabular}{ p{0.5cm}ccc ccc ccc }
 \hline
\hline\\[-4.5ex]
& \multicolumn{3} {c}{$n=500$ }&\multicolumn{3} {c}{$n=1000$ }&\multicolumn{3} {c}{$n=1500$ }\\
 \cmidrule(lr){2-4}\cmidrule(lr){5-7}  \cmidrule(lr){8-10}\\[-4.5ex] 
&\text{East}& \text{West} & \text{Ger} &\text{East}& \text{West} & \text{Ger} &\text{East}& \text{West} & \text{Ger}   \\ 
 \cmidrule(lr){2-2}\cmidrule(lr){3-3}  \cmidrule(lr){4-4} \cmidrule(lr){5-5}
 \cmidrule(lr){6-6}\cmidrule(lr){7-7}  \cmidrule(lr){8-8} \cmidrule(lr){9-9}\cmidrule(lr){10-10} \\[-4.5ex]  \\[-4.5ex] 
\text{S 1}&0.034& 0.134& 0.132&0.038& 0.329& 0.303& 0.026& 0.553& 0.535\\[.5ex] 
\text{S 2}&0.063& 0.204& 0.164&0.090& 0.517& 0.479& 0.099& 0.755& 0.673\\[.5ex] 
\text{S 3}&0.050& 0.136& 0.094&0.026& 0.353& 0.339& 0.030& 0.551& 0.529\\[.5ex] 
\text{S 4}&0.092& 0.198& 0.158&0.086& 0.449& 0.461& 0.104& 0.745& 0.661\\[.5ex] 
\text{S 5}&0.089& 0.429& 0.387&0.276& 0.880& 0.775& 0.507& 0.966& 0.948\\[.5ex] 
\hline
 \hline
\end{tabular}
\caption*{\spacingset{1.0}\footnotesize Incomes. Shown are  the rejection rates of size $n$ of the specification S1--S5. The number of sub-samplings is $501$ and the critical values are calculated at a significance level of $5\%$  and  for $\tau\in\{0.1,0.2,\ldots,0.9\}$.}
\label{metatest}\vspace{-.75cm}
\end{table} 
\\
From this table we make two observations: First, it can be observed that the model in which age has  a completely quadratic influence on income (S5) provides the worst fit. Also the models with either a completely linear influence or a linear influence in the upper quantiles (S2 and S4) are worse than the models in which the influence is quadratic in the upper quantiles and linear in the lower ones (S1 and S3). Second, the model fits are in general much better for East Germany than for West Germany and for the whole country, whose rejection rates can be interpreted as the weighted average of the two rejection rates. For example, for $n=1500$, the rejection rate of S1 and S3 are even lower than $5\%$ for East Germany, whereas they are larger than $50\%$ for West Germany. This indicates that the conditional income distributions differ significantly between East and West Germany.
\\ 
Finally, Figure \ref{fig1_quant} visualizes the estimated quantiles at $\tau=0.1, 0.5, 0.9$ (from left to the right) and provides further indications of when age might have a quadratic or linear effect. Shown are the results for West (red), East (green) and entire Germany (blue). 
\begin{figure}[htbp]
	\centering
	\caption{\footnotesize{Income quantiles for East/West and entire Germany as functions of $age$}}
	\includegraphics[width=.95\textwidth]{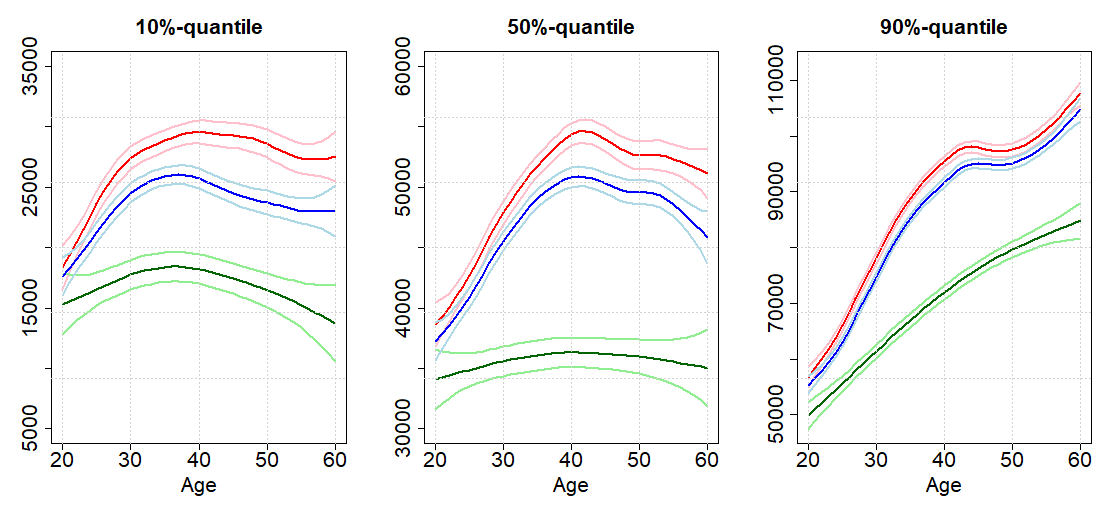}
	\caption*{\spacingset{1.0}\footnotesize Incomes. Figures show the penalized conditional quantile estimates for West (red), East (green) and entire Germany (blue) at $\tau=0.1, 0.5, 0.9$.  The lines shown in lighter colors represent the $95\%$ confidence intervals.}
	\label{fig1_quant}\vspace{-.5cm}
\end{figure}
Overall, our results are in line with the findings of other studies. Based on the different structure of the conditional qfs and  rejections rates for different specifications significant structural differences between East and West Germany can still be assumed \citep{kluge}.

\subsection{Interaction Effects in Modelling Australian Electricity Prices}\label{sec:edata}
In this section, we apply the specification test $S_n^{CM,S}$  to  electricity data from the Australian national electricity market (NEM) in 2019. The NEM is a wholesale market, where generators, distributors and third party participants bid for sale and purchase of electricity one day ahead of transmission \citep[][]{ignatieva2016modeling,smith2018econometric}. We consider  hourly market-wide price $P_i$ from January 1, 2019 to December, 31, 2019, which yields $n=8760$ observations. The market-wide price $P_i$ is the demand-weighted average price across the five regions (\url{www.aemo.com.au}). We correct for the three main drivers of the electricity spot price distribution, namely day of the year $x_{1}$, time of day $x_{2}$ and total market demand $x_{3}$, which is the sum of demand across the five regions in the NEM. Following \citet{smith2019bayesian} we thus choose a regression approach for the electricity data from the Australian NEM even if the problem could be addressed by a time series approach. For convenience, we scale each covariate to the unit interval. 
\\
Our main purposes are to identify i) potential interactions between the covariates on different quantiles of the electricity spot price distributions ii) to statistically investigate if the impact of the covariates $x_1, x_2$ and $x_3$ varies for distinct quantiles and iii) to test which (interaction) effects are statistically significant. In contrast to the previous application in Sec.~\ref{subsec51}, it is not clear a priori how to optimally determine a functional relationship between the three covariates $x_{1}$, $x_{2}$ and $x_{3}$ for distinct quantiles $\tau\in(0,1)$. Therefore, the functional relationship for different quantiles is modeled very flexibly by a spline approach. We employ trivariate P-splines (tensor product B-splines) as proposed by \citet{eilers1996flexible} which combine a multivariate B-spline basis, with a discrete penalty on the basis coefficients. 
\\
In order to investigate our main purposes i)--iii), we assume that the data generating process can be represented by one of the eight different specifications S6--S13. To increase the readability, the notation is geared to the implementation in R, i.e.~$s(\cdot,\tau)$ models the marginal P-spline and $ti(\cdot,\tau)$ solely the interaction effect at the quantile $\tau$. For example, S6 describes a P-spline for the three covariates $x_1, x_2, x_3$ represented by the marginal main effects $s(x_1,\cdot),$ $s(x_2,\cdot),$ $s(x_3,\cdot)$, their mutual bivariate interactions $ti(x_1,x_2,\cdot),\ ti(x_1,x_3,\cdot), \ ti(x_2,x_3,\cdot)$ and their mutual trivariate interaction $ti(x_1,x_2,x_3,\cdot)$. In contrast, specification S7 does not incorporate any interactions between the covariates and thus models the marginals effects only. Specification S12 describes a P-spline that models the marginals and bivariate interaction effects within the $0.25$ and $0.75$-quantile. Specifically, we define

\vspace{-.25cm}
\spacingset{1.4}
\begin{align*}
\begin{split}
\text{S6:} \qquad &F_{Y\mid X}^{-1}(\tau\mid x_1,x_2,x_3):=s(x_1,\tau)+s(x_2,\tau)+s(x_3,\tau)\\&\quad\qquad\qquad\qquad\qquad+ti(x_1,x_2,\tau)+ti(x_1,x_3,\tau)+ti(x_2,x_3,\tau)+ti(x_1,x_2,x_3,\tau)\\
\text{S7:} \qquad &F_{Y\mid X}^{-1}(\tau\mid x_1,x_2,x_3):=s(x_1,\tau)+s(x_2,\tau)+s(x_3,\tau)\\
\text{S8:} \qquad &F_{Y\mid X}^{-1}(\tau\mid x_1,x_2,x_3):=s(x_1,\tau)+s(x_2,\tau)+s(x_3,\tau)+ti(x_1,x_3,\tau)+ti(x_2,x_3,\tau)\\
\text{S9:} \qquad &F_{Y\mid X}^{-1}(\tau\mid x_1,x_2,x_3):=\begin{cases}
S6\ , &\text{ if  } 0.25<\tau< 0.75\\
S7\ , &\text{ otherwise}
\end{cases}
\end{split}
\end{align*}
\begin{align*}
\begin{split}
\text{S10:} \qquad &F_{Y\mid X}^{-1}(\tau\mid x_1,x_2,x_3):=\begin{cases}
S6\ , &\text{ if  } \tau\leq 0.25\\
S7\ , &\text{ otherwise}
\end{cases}\\
\text{S11:} \qquad &F_{Y\mid X}^{-1}(\tau\mid x_1,x_2,x_3):=T4\\
\text{S12:} \qquad &F_{Y\mid X}^{-1}(\tau\mid x_1,x_2,x_3):=\begin{cases}
S6-ti(x_1,x_2,x_3,\tau)\ , &\text{ if  } 0.25<\tau< 0.75\\
S7\ , &\text{ otherwise}
\end{cases}\\
\text{S13:} \qquad &F_{Y\mid X}^{-1}(\tau\mid x_1,x_2,x_3):=\begin{cases}
S6-ti(x_1,x_2,\tau)-ti(x_1,x_2,x_3,\tau),&\text{ if  } 0.25<\tau< 0.75\\
S7\ , &\text{ otherwise}
\end{cases}\\
\end{split}
\end{align*}
\spacingset{1.8}
Similar to the previous section, all estimations were carried out using the  \texttt{qgam} package in R. As before we used REML to optimize the smoothing parameters. Due to the extreme skew in electricity prices, we follow previous authors and set $Y_i=log(P_i+17)$. This avoids a negative dependent variable $Y_i$ since the minimum observed price in our data is $-\$15.78$. For the application of our test, we set $n\in\{500,1000,2000\}$ and $\tau\in\{0.02,0.04,\ldots,0.98\}$. 
 The number of sub-samples is $501$ and the critical values are calculated at a significance level of $5\%$. To ensure comparability of rejection rates for different $n$ and since we include multivariate interaction effects (cf.~S6 and S8--S13), we set the number of knots to $5$.  The rejection rates of the specification test $S_n^{CM*}$ are listed in Table \ref{tabstrom}. 
 \vspace{-.35cm}
\begin{table}[htbp]
\spacingset{1.5}\footnotesize\centering
\caption{\footnotesize Rejection rates of the test statistic $S_n^{CM^*}$}
\begin{tabular}{ p{1.0cm}cc ccc cc c}
 \hline
\hline\\[-4.5ex]
&\text{S6}& \text{S7}&\text{S8}&\text{S9}  & \text{S10}  & \text{S11} & \text{S12}  & \text{S13}    \\ 
\hline\\[-4.5ex]
\text{n=500}  &0.050 & 0.115& 0.065& 0.086& 0.043& 0.086& 0.058& 0.079 \\[.ex] 
\text{n=1000} &0.089 & 0.338& 0.300& 0.178& 0.185& 0.135& 0.218& 0.224 \\[.ex]
\text{n=2000} &0.228 & 0.811& 0.748& 0.256& 0.237& 0.445& 0.713& 0.764\\[.ex]
\hline
 \hline
\end{tabular}
\caption*{\footnotesize Electricity prices. The table shows the sub-sample rejection rates of size $n$ of the specification S6--S13.}
\label{tabstrom}\vspace{-.6cm}
\end{table} 
From this table we make four observations. First, it can be observed that the rejection rates increase as $n$ increases, which is plausible as our specification test is consistent and also small deviations from the null hypothesis are detected for large sample sizes. In addition, an increase in the rejection rates as $n$ increases could be due to possible structural breaks.  Second, based on the rejection rates for S8--S13 at $n=500$, interaction effects seem to have a significant impact, especially in the lower quantile, i.e.~at $\tau \leq 0.25$. This is particularly reflected in the comparison of the specifications S9 and S10, which differ in the modeling of the upper quantile ($\tau\geq 0.75$) but show similar rejection rates. 
Third, 
we can conclude from the specifications S12 and S13 that the interaction between the day of the year ($x_1$) and the daytime ($x_2)$ has no significant impact to the log electricity prices.
Fourth, specification S7, however, which does not incorporate interaction effects, is rejected at all sample sizes.
\\ 
Figure \ref{fig_S6_Visuals} shows the decomposition of the main and interaction effects at the 90\% quantile at 6:00 p.m.~using  S6. Since the contour lines in the second and third panel  show the presence of interactions between demand and day, we conclude that the relation between the three covariates cannot be fully captured by product interactions based on univariate splines. In addition, different day-demand combinations have a different impact on the market wide price $P_i$. A similar graphical analysis additionally reveals this behavior for the 10\% quantile (see Supplement \ref{app:electricity}). This observation is also confirmed by the higher rejection rates of our specification test when using univariate splines rather than bivariate tensor product B-splines (cf.~Table \ref{tabstrom} and Table \ref{tabstrom_prod} in  Supplement \ref{app:electricity}). 
Overall, we conclude that for a thorough specification of the Australien NEM mutual interaction effects are important. Particularly, there seems to be a complex dependence structure in the lower quantile ($\tau\leq 0.25$)  of log electricity prices, which can be captured by means of the mutual interaction effects. However, the interaction between day of the year and daytime is negligible here.  This might be important for risk management purposes.
\begin{figure}[htb]
	\centering
	\caption{\footnotesize{Estimated main and interaction effects at the 90\% quantile at 6:00 p.m.}}
	\includegraphics[width=.98\textwidth]{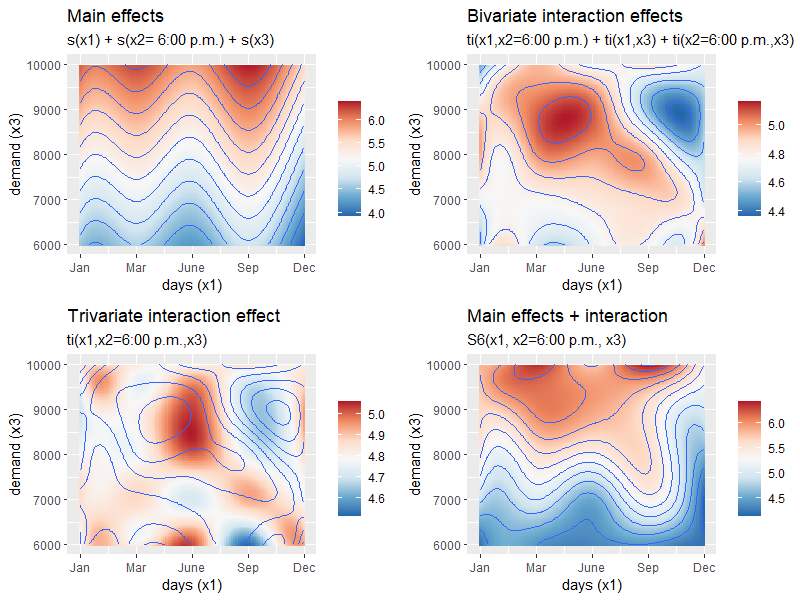}
	\caption*{\spacingset{1.0}\footnotesize Electricity prices. Figures depict the estimated effects of the three covariates on the 90\% quantile of the Australian NEM hourly electricity price distribution for 2019, where the time of the day (x2) is set to 06:00 p.m. The estimation is based on the model specification S6. The first panel (upper left) shows the sum of the univariate main effects of days (x1), x2 and total market demand (x3). The second and third panel illustrate the bivariate and trivariate interaction effects. The overall effect is depicted in the last panel.}
	\label{fig_S6_Visuals}\vspace{-.5cm}
\end{figure}

\section{Conclusion}\label{sec:conclusion}

In this paper, we derived and tested new specification tests for parametric and semi-parametric quantile regression models, which allow the covariates to vary over quantiles in a flexible non-linear way. To improve finite sample properties in the parametric model framework, we replace the non-parametric ecdf by an estimator that is based on an estimate of the quantile regression function using penalized splines. Our MC study illustrates that the proposed method has superior test properties compared with several existing benchmarks from the literature. 
We have illustrated this in two famous examples on income inequality and electricity spot prices: The (nonlinear) effect of age on the income distribution is a well-known example. A detailed investigation of the conditional income distributions between East and West Germany using the M-M decomposition reveals that still income differences between the regions in Germany are present, even more than two decades after the reunification. Similarly, modelling and predicting electricity spot prices is a common issue in economics. We treat the problem in a semi-parametric framework and reveal the importance of interaction effects between demand and time variables, particularly for lower quantiles of the price distributions.
\\
We believe our test statistics make an important contribution in the specification testing literature  since nonlinear or even more complex functional forms of covariates are omnipresent in many applications.

\spacingset{1.3}
\footnotesize
\setlength{\bibsep}{0pt plus 0.1ex}
\bibliography{Bibliography-MM-MC}

\newpage
\spacingset{1.8}
\setcounter{page}{1}
\bigskip
\begin{center}
{\large\bf SUPPLEMENTARY MATERIAL}\label{sec::suppApp}
\end{center}
\begin{description}
\item[Part \ref{sec::proof}:] Proofs and derivations.
\item[Part \ref{secMC}:] Monte Carlo simulation study for $S_n^{CM}$, $S_n^{CM^*}$ and $S_n^{CM,S}$.
\item[Part \ref{app:electricity}:] Additional figures and results to the empirical application on electricity prices of the Australian NEM of the manuscript.
\end{description}

\renewcommand\thesection{\Roman{section}}
\setcounter{table}{0}
\setcounter{figure}{0}
\renewcommand\thetable{\Roman{table}}
\renewcommand\thefigure{\Roman{figure}}
\setcounter{section}{0}
\section{Proofs}\label{sec::proof}

\subsection{Proof of Theorem \ref{t1}}\label{sec::A1}
In order to maintain readability we omit the index $Y\mid X$ for the conditional cdf $F$. To prove Theorem \ref{t1}, we first derive and prove three auxiliary results. Therefore, we define the following three processes for $(y,x)\in\R^{K+1}$ and $(\theta,\tau)\in\Theta\times\T$:
\begin{align}
\nu_n(y,x)&:=\sqrt{n}\left(\hat{F}_n(y,x)-F(y,x)\right)\label{pro1}\\
\gamma_n(\te,\tau)&:=\sqrt{n}\left(\hat{G}_n(\te,\tau)-G(\theta,\tau)\right)\label{pro2}\\
\nu_n^0(y,x)&:=\sqrt{n}\left(\hat{F}_n(y,x,\hat{\te}_n)-F(y,x,\theta_0)\right)\label{pro3}.
\end{align}
Let $\ell^\infty$  denote the set of all uniformly bounded real functions.
\begin{lemma}\label{l1}
Assume Assumption \ref{ass1} holds. For the processes \eqref{pro1} and \eqref{pro2} it holds under the null, that 
\begin{align*}
(\nu_n,\gamma_n)\Rightarrow \tilde{\G}:=(\G_1,\tilde{\G}_2) \text{ in } \ell^\infty(\Supp\times \Theta\times \T),
\end{align*}
where $\tilde{\G}$ is a tight bivariate mean zero Gaussian process.
\end{lemma}
\begin{proof}
First, we notice that the Donsker property is conserved under the union of Donsker classes. Hence, $\nu_n$ and $\gamma_n(\theta,\tau)$ are $F_{YX}$- Donsker for all $\theta\in\B$ and $\tau\in\T$ with limiting processes $\G_1$ and $\tilde{\G}_2$, respectively. Since arbitrary linear combinations of $\nu_n$ and $\gamma_n$ are Lipschitz and thus Donsker \citep[see][Example 29.20]{vaart_1998}, we conclude by the Cram\'{e}r-Wold theorem that $(\nu_n,\gamma_n)$ converge in distribution to $\tilde{\G}$.
\end{proof}
Before we prove the next lemma we slightly generalize Lemma E.3 from \citet{cher:2013} for our purposes. This modification summarized in the following corollary states conditions under which a $Z$-estimation process satisfies the functional delta method for Gaussian processes. 
\begin{corollary}\label{coro}
Let Assumption \ref{ass1} \ref{A11}.)--\ref{A14}.) be satisfied and $\sqrt{n}\left(\hat{G}_n-G\right)\Rightarrow \tilde{\G}_2$ in $\ell^\infty(\Theta\times I_l)$ for all $l=1,\ldots,L$, where $\tilde{\G}_2$ is a Gaussian process with a.s. uniformly continuous paths on $\Theta\times I_l$, $l=1,\ldots,L$. Further, we assume that the estimator $\hat{\theta}_n(\tau)$ is an approximate $Z$-estimator \eqref{Zestimator} for all $\tau\in I_l$ with $l=1,\ldots,L$. Then 
\begin{align*}
\sqrt{n}\left(\hat{\theta}_n(\cdot)-\theta_0(\cdot)\right)&=-\dot{G}^{-1}_{\theta_0(\cdot),\cdot}\left[ \sqrt{n}(\hat{G}_n-G)(\theta_0(\cdot),\cdot)\right]+o_P(1) \\
&\Rightarrow  -\dot{G}^{-1}_{\theta_0(\cdot),\cdot}\left[ \tilde{\G}_2(\theta_0(\cdot),\cdot)\right] \in \ell^{\infty}(\T). 
\end{align*}
If Assumption \ref{ass1} \ref{A15}.) also holds true, then the paths $\tau\mapsto -\dot{G}^{-1}_{\te_0,\tau}\left[\tilde{\G}_2(\te_0,\tau)\right]$ are a.s.~uniformly continuous on $\T$.
\end{corollary}

\begin{proof}
The intersection of $I_{l_1}$ and $I_{l_2}$ is a singleton by assumption for $l_1\neq l_2$.  Thus, the set of possible discontinuities is a null set with respect to the Lebesgue measure. Hence, the limiting process $\tilde{\G}_2$ is a.s. continuous on $\Theta\times \T$ with respect to the Euclidean metric. Further we notice, that by assumption the decomposition of the unit interval is finite. Consequently, the property of uniformity is also applicable to the finite union of compact sets. Hence, the conditions of Lemma E.3 in \citet{cher:2013} are fulfilled.
\end{proof}

\begin{lemma}\label{l2}
Let either the null hypothesis or a fixed alternative and Assumptions \ref{ass1} 
 be true. Then 
\begin{align*}
(\nu_n,\nu_n^0)\Rightarrow\G:=(\G_1,\G_2)\text{ in } \ell^\infty(\Supp\times \Supp),
\end{align*}
where $\G_1$ is the limiting tight bivariate mean zero Gaussian process of $\nu_n$ and 
\begin{align*}
\G_2:=\int F(y\mid x^*)\one_{\{x^*\leq x\}}d\G_1(\infty, x^*)+\int \G_2^+(y,x^*)\one_{\{x^*\leq x\}}dF_X(x^*)
\end{align*}
with $\G_2^+:=-\dot{F}(y\mid x,\te_0)\left[\left[\dot{G}(\te_0(\cdot),(\cdot)\right]^{-1}\tilde{\G}_2(\cdot)\right]$.
\end{lemma}

\begin{proof}
Under either the null hypothesis or a fixed alternative, it follows by standard arguments from Lemma \ref{l1} and Corollary \ref{coro} that 
\begin{align*}
\sqrt{n}\left(\hat{F}_n(\cdot,\cdot)-F(\cdot,\cdot),\hat{\te}_n(\cdot)-\te_0(\cdot)\right)\Rightarrow \left(\G_1(\cdot,\cdot),- \dot{G}^{-1}_{\te_0(\cdot),\cdot}(\tilde{\G}_2(\te_0(\cdot),\cdot)\right)\text{ in }\ell^{\infty}(\Supp)\times\ell^{\infty}(\T).
\end{align*}
Next, it follows from the Hadamard differentiability (cf. Assumption \ref{ass1} \ref{A17}.)) that 
\begin{align*}
\sqrt{n}\left(\hat{F}_n(y\mid x,\hat{\theta}_n)-F(y\mid x,\theta_0)\right)\Rightarrow -\dot{F}(y\mid x,\theta_0)\left[\dot{G}^{-1}_{\te_0(\cdot),\cdot}(\tilde{\G}_2(\te_0(\cdot),\cdot)\right]=:\G_2^+(y,x).
\end{align*}
The statement of the lemma then follows directly from the Hadamard derivative $\dot{\phi}$ of the mapping  
\begin{align*}
\phi(A,B)[x^*]:= \int A(\cdot, x^*)\one_{\{x^*\leq \cdot \}}d B(x^*)
\end{align*}  
given by 
\begin{align*}
\dot{\phi}_{\alpha,\beta}(A,B)[x^*]=\int A(\cdot, x^*)\one_{\{x^*\leq \cdot \}}d\beta(x^*)+\int\alpha(\cdot,x^*)\one_{\{x^*\leq \cdot \}} dB(\cdot, x^*)
\end{align*} and the functional delta method. In particular, for the second component $\G_2$ of the joint limiting process, we have
\begin{align*}
\G_2(y,x)=\int \G_2^+(y,x^*)\one_{\{x^*\leq x\}}dF_X(x^*)+\int F(y\mid x^*)\one_{\{x^*\leq x\}}d\G_1(\infty, x^*).
\end{align*}
\end{proof}

\begin{proof}[of Theorem \ref{t1}] 
We start with the first statement of Theorem \ref{t1}. Under the null hypothesis it holds that $\hat{F}_n(y,x)=F(y,x,\theta_0)+o_p(1)$ for all $(y,x)\in\Supp$. By linearity, we have 
\begin{align*}
S_n^{CM}&=\sqrt{n}\int\left(\hat{F}_n(y,x)-\hat{F}_n(y,x,\hat{\theta})\right)d\hat{F}_n(y,x)\\
&=\int\left(\nu_n(y,x)-\nu_n^0(y,x)\right)^2dF(y,x)+\int\left(\nu_n(y,x)-\nu_n^0(y,x)\right)^2d\left(\hat{F}_n(y,x)-F(y,x)\right).
\end{align*}
From Lemma \ref{l2} we know that $(\nu,\nu_0)\Rightarrow (\G_1,\G_2)=\G$, where $\G$ is a tight bivariate mean zero Gaussian process. Applying the continuous mapping theorem and the Donsker class property yields
\begin{align*}
S_n^{CM}=\int \left(\G_1(y,x)-\G_2(y,x)\right)^2dF(y,x)+o_p(1)
\end{align*}
which claims the statement. 
\\
To show part $ii.)$, we use the fact that under any fixed alternative  $P(F(y,x)\neq F(y,x,\te_0)>0$ due to construction of the alternative hypothesis in $\eqref{null3}$. Thus, 
\begin{align*}
S_n^{CM}=\int \left(\nu_n(y,x)-\nu_n^0(y,x) + \sqrt{n}(F(y,x)-F(y,x,\theta_0)\right)^2d F(y,x)+ {o}_P(1)=\mathcal{O}_P(n),
\end{align*}
which implies that $S_n^{CM}$ is greater than any fixed constant $\varepsilon>0$ and hence, the probability that $S_n^{CM}$ is greater than any $\varepsilon>0$ tends to $1$.
\end{proof}
\vspace{-1cm}
\subsection{Proof of Theorem \ref{t2}}\label{sec::A2}
The proof is shown for $P(X,\tau)=P(X)$. In case of quantile dependent regressors, standard arguments those as given in the proof of Theorem \ref{t1} apply. To prove Theorem $\ref{t2}.i)$, we consider the parametric and the semi-parametric model with increasing dimension. 
The steps from the proof of Theorem \ref{t1} are applied analogously replacing $v_n(y,x)$ and $v_n^0(y,x)$ from \eqref{pro1} and \eqref{pro3} by $v_n^{0}:=a_n\left(\hat{F}_n(y,x,\hat{\te}_n)-F(y,x,\theta_0)\right)$ and  $v_n^{0,S}:=a_n^*\left(\hat{F}^S_n(y,x,\hat{\te}_n)-F(y,x,\theta_0)\right)$ with $a_n={\sqrt{n}}/{\rVert P(\xo)\lVert}$ and $a_n^*={\sqrt{n}}/{\rVert B(\xo)\lVert}$, respectively. By Theorem 1, Corollary 1 in \cite{belloni2019} and Theorem \ref{t1} we have that 
\begin{align*}
   (v_n^0, a_n(\hat{\theta}_n(\cdot)-\theta_0(\cdot)))\Rightarrow \left(\G_2(\cdot,\cdot),- \G_2^{S_M}\right), 
\end{align*}
where $\G_2^{S_M}:=-\dot{G}^{-1}_{\te_0(\cdot),\cdot}\left[\tilde{\G}^{S_M}_2(\te_0(\cdot),\cdot)\right]$. Together with the Hadamard differentiablity in Assumption \ref{ass1} and if 
\begin{align*}
    H(\tau_1,\tau_2,P(\xo)):=\lim\limits_{n\to\infty}\rVert P(\xo)\lVert^{-2}P(\xo)^\top J_m^{-1}(\tau_1)\E[P(\xo)P(\xo)^\top]J_m^{-1}(\tau_2)P(\xo)(\tau_1\wedge\tau_2-\tau_1\tau_2)
\end{align*}
exists for any $\tau_1,\tau_2\in\mathcal{T}$, we know from \cite{volgushev2019distributed} Theorem 2.1 and Corollary 4.1 that for any fixed $\xo$ and initial estimator $\hat{F}^{-1}_{n}(\cdot\mid \xo)$ the expression $a_n^*(\hat{F}_{n}(\cdot\mid \xo)-F_{Y\mid X}(\cdot\mid \xo))\Rightarrow -f_{Y\mid X}(\cdot\mid \xo)\tilde{\G}_2^{S_M}$, where $\tilde{\G}_2^{S_M}$ is a centered Gaussian process with covariance function $H(\tau_1,\tau_2,P(\xo))$. 
\\
In order to show that $(v_n^0,v_n^{0,S})\Rightarrow \G_S:=(\G_2,\G_2^{S_M})$ we use the Hadamard differentiability of the mapping $\phi(A,B)[x^*]:= \int A(\cdot, x^*)\one_{\{x^*\leq \cdot \}}d B(x^*)$  the functional delta method as stated. The continuous mapping theorem completes the proof. 
Part $ii.)$ can be proved analogously to part $ii.)$ of Theorem \ref{t2}.  Moreover,
\begin{align*}
Cov(\G_2(y_1,x_1),&\G_2^{S_M}(y_2,x_2))=\\
&\lim_{n \rightarrow \infty} n\ Cov\left(\hat{F}_n(y_1,x_1,\hat  {\te}_n)-F(y_1,x_1,  {\te}_n),\hat{F}_n^{S_M}(y_2,x_2,\hat{\theta}_n)-F(y_2,x_2,\theta)\right),
\end{align*}
where the true functional vector $\te(\tau)$ depends on $n$ for all $\tau\in\mathcal{T}$.

\subsection{Proof of Theorem \ref{T3}}\label{sec::A3}

In order to prove Theorem \ref{T3} we present the bootstrap version of Lemma \ref{l1} as an auxiliary result. 
\begin{lemma}\label{l3}
Let Assumption \ref{ass1} be true. We define the bootstrap version of the empirical processes \eqref{pro1} and \eqref{pro3} 
\begin{align}
\begin{split}
\nu_{n,B}(y,x)&:=\sqrt{n}\left(\hat{F}_{n,B}(y,x)-\hat{F}_n(y,x,\hat{\te}_n)\right)\\
\nu_{n,B}^0(y,x)&:=\sqrt{n}\left(\hat{F}_{n,B}(y,x,\hat{\te}_n)-\hat{F}_n(y,x,\hat{\te}_n)\right).
\end{split}\label{pro1b}
\end{align}
Then it holds under either the null or a fixed alternative hypothesis that 
\begin{align*}
\left(\nu_{n,B},\nu_{n,B}^0\right)\Rightarrow \G_b,
\end{align*}
where $\G_b:=(\G_{b1},\G_{b2})$ is a tight bivariate mean zero Gaussian process whose distribution function coincides with that of the process $\tilde{\G}$ in Lemma \ref{l1}.
\end{lemma}
\begin{proof}
This follows from Lemma \ref{l1} and the functional delta method for the bootstrap \citep{rothewied:2013}.
\end{proof}
\begin{proof}[Proof of Theorem \ref{T3}]
To prove part $i.)$ of Theorem \ref{T3}, let $c(\alpha)$ be the true critical value satisfying $P(S_n^{CM}>c(\alpha))=\alpha+o_P(1)$. Then it follows from Lemma \ref{l3} that $\hat{c}_n(\alpha)=c(\alpha)+o_P(1)$. This implies that $S_n^{CM}$ and $\tilde{S}_n:=S_n^{CM}-(\hat{c}_n(\alpha)-c(\alpha))$ converge to the same limiting distribution as $n$ tends to infinity. Hence, $P(S_n^{CM}>\hat{c}_n(\alpha))=\alpha+o_P(1)$ as claimed.
To prove part $ii.)$, we deduce from Lemma \ref{l3} that the bootstrap critical values are bounded in probability under fixed alternatives. Thus, for any $\varepsilon>0$, there is an $N(\epsilon)$ such that $P(\hat{c}_n(\alpha)>N(\varepsilon))<\varepsilon+o_P(1)$. By Kolmogorv axioms we obtain
\begin{align*}
P(S_n^{CM}\leq \hat{c}_n(\alpha))&=P(S_n^{CM}\leq \hat{c}_n(\alpha),S_n^{CM}\leq N(\epsilon))+P(S_n^{CM}\leq \hat{c}_n(\alpha),S_n^{CM}> N(\epsilon))\\
&\leq P( S_n^{CM}\leq N(\epsilon))+P( S_n^{CM}> N(\epsilon))\\
&\leq \varepsilon + o_P(1),
\end{align*}
where the last inequality can be deduced from Theorem \ref{t1} $ii.)$. 
\end{proof}

\newpage

\section{Monte Carlo Simulation Study}\label{secMC}
Section~\ref{sub41} contains a comprehensive MC simulation study for the test statistics $S_n^{CM}$ and $S_n^{CM^*}$, where the spline part in the latter test statistic is modelled by a penalized B-spline. Wherever it is possible, we also compare our results with existing benchmark tests, for instance, those given in \cite{koenkerxiao} ($KX$), \cite{cher:2002} ($CH$) and \cite{rothewied:2013} ($RW$). Note, in case of quantile-independent covariates, $RW$ is a special case of our proposed test $S_n^{CM}$. In Section~\ref{sec:MCsemi}, we examine power and size properties for the semi-parametric model specification test $S_n^{CM,S}$ where possible interacting covariates are modelled by a tensor product. In Section~\ref{app:mcstudy}, we examine power and size properties for the semi-parametric model specification test $S_n^{CM,S}$ with univariate product interacting covariates. 

\subsection{MC Simulation Study for \texorpdfstring{$S_n^{CM^*}$}{Lg}}\label{sub41}
In this subsection, we show that our test $S_n^{CM^*}$ holds the size level and has superior power properties compared with $S_n^{CM}$ by means of the twelve different data generating processes (DGPs) based on i.i.d.~data $\{(y_i,x_i)\mid 1\leq i\leq n\}$ for $n\in\lbrace 30,50,100,300,500,1000,2000\rbrace$. The different DGPs cover location shift models (LS) and location-scale shift models (LSS) including heteroscedastic errors, both, in a univariate and multivariate setting. In order to assess the quality and validity of our proposed test against existing procedures, we benchmark against the tests of \citet{koenkerxiao}, \citet{cher:2002} and \citet{rothewied:2013} where comparisons are possible (DGPs 1--9). Finally, we also consider  linear models and show that our test detects even weakly misspecified models well. 

For the definition of the twelve DGPs we introduce the following variables: Let $x_0\in U(0,2\pi) $, $x_1\sim Bin(1,0.5)$, $x_2\sim N(0,1)$,  $x_3\in U(0,1)$ , $x_4\in\chi^2(1)$, $u\sim N(0,1)$, $w\sim N(0,0.1)$, $v=(1-2x_1)\cdot v_2^*\cdot 8^{-0.5}$ with $v_2^*\sim \chi^2(2)$, where $Bin(\cdot, \cdot)$, $N(\cdot,\cdot)$, $U(\cdot, \cdot)$ and $\chi^2(\cdot)$ are  Binomial,   Gaussian,   uniform and   chi-square distributions, respectively. 
\paragraph{Data Generating Processes}
DGPs 1--3 
represent the univariate case with one covariate and additive noise. Hereby, DGP $1$ describes a simple LS model, DGP $2$ a more complex LSS model with a linear regressor and, finally, DPG $3$ generates a quadratic LSS model. The multivariate case 
is specified by the DGPs 4--8 that are from \citet{rothewied:2013} and DGP $9$ from \citet{cher:2002}. Here, DGP $4$ is a simple multivariate LS model with normally distributed errors. DGP $2$ is again a simple LS model, but now the errors follow a mixture of a ``positive'' and ``negative'' $\chi^2$ distribution with two degrees of freedom (normalized to have unit variance). DGPs 6--8 are multivariate LSS models where the level of heteroscedasticity increases. DGP $9$ is considered in order to compare our proposed testing procedure with those provided in \citet{cher:2002} and \citet{koenkerxiao}. When $\gamma_1=0$ DGP $9$ is a LS model, otherwise it is a LSS model. DGPs 10--12 
are processes in which the functional form appears predominantly linear. DGP $10$ is implemented by modeling the lower $50\%$-quantile linearly, while the upper $50\%$-quantile is modeled quadratically. Due to the quantile dependence of the regressors,  DGP $10$ cannot be correctly tested with previous tests but with our test $S_n^{CM^\ast}$. DGP 11--12 are appearing mainly linear in the interval $[0,1]$ and exhibit nonlinear growth only at values close to $1$. Assuming a linear model, DGPs of the form 10--12 often impede the detection of misspecification.

\paragraph{Estimation and Further Settings}
Computations have been carried out using the R package \texttt{cobs} \citep{cobs,cobs2}. In what follows, $\hat{F}_n^S$ is modeled by a  B-spline of second order with penalty term $\lambda=1$ and $\sqrt{n}$ knots evaluated for $\tau\in\{0.1,0.2,\ldots,0.9\}$, meeting monotonicity assumptions. The number of MC repetitions is equal to $701$ with $500$ bootstrap replications. The significance level is $0.05$. 

\begin{align}\label{dgp_quant1}
\text{DGP 1:}\ &f_1(x_0):=\frac{x_0}{4}+ 1+ u,\qquad\qquad\!\!\qquad
\text{DGP 2:}\ f_2(x_0):=\frac{x_0}{4} + 1+ u\cdot x_0\nonumber\\
\text{DGP 3:}\ &f_3(x_0):=\frac{x_0^2}{4} + 1+ u\cdot x_0^2,\qquad\qquad
\text{DGP 4:}\ f_4(x_1,x_2):=x_1+x_2+u\nonumber\\
\text{DGP 5:}\ &f_5(x_1,x_2):=x_1+x_2+v,\qquad\qquad
\text{DGP 6:}\ f_{6}(x_1,x_2):=x_1+x_2+(\frac{1}{2}+x_1)u\nonumber\\
\text{DGP 7:}\ &f_{7}(x_1,x_2):=x_1+x_2+(\frac{1}{2}+x_1+x_2^2)^{0.5}u \\
\text{DGP 8:}\ &f_{8}(x_1,x_2):=x_1+x_2+\frac{1}{5}(\frac{1}{2}+x_1+x_2^2)^{1.5}u\nonumber\\
\text{DGP 9:}\ &f_{9}(x_3):=x_3+(1+\gamma_1\cdot x_2)u\nonumber\\
\text{DGP 10:}\ &f_{10}(x_3):=\begin{cases}
\frac{x_3^2}{4}+1+\frac{\epsilon\cdot x_3^2}{2}, &\text{ if } \tau\geq 0.5\nonumber\\
\frac{-x_3^2}{4}+1+u\cdot x_3, &\text{ otherwise}\nonumber
\end{cases}&&\nonumber\\
\text{DGP 11:}\ &f_{11}(x_3):=\sin\left(-\frac{\pi}{2}+x^3_3\right)+w,\qquad\!\!
\text{DGP 12:}\ f_{12}(x_3):=e^{f_5(x_3)}\nonumber
\end{align}

\paragraph{Benchmark Tests}
In order to illustrate the performance of our test, we draw comparisons to common test procedures in the scope of quantile regression. The test proposed in \citet{koenkerxiao} ($KX$), which is based on the Khmaladze transformation, which in turn refers to the Doob-Mayer decomposition of martingales, provides the starting point for quantile regression specification tests. We also consider the enhancement proposed in \citet{cher:2002} ($CH$) and compare our test with $RW$. The aforementioned tests are characterized as follows: 
\begin{itemize}
\item[$\bullet$] The $KX$-test models the conditional qf parametrically by assuming a LS or a LSS model. The regressors are fixed for all quantiles considered and the estimation of non-parameteric sparsity and score functions are required \citep{cher:2002}.
\item[$\bullet$] In order to avoid the latter, $CH$ employs a resampling testing procedure based on $KX$ that results in better power and accurate size. However, this tests still assumes a fully parametrized model under the null hypothesis with quantile-independent regressors. 
\item[$\bullet$] $RW$ propose a testing procedure for a wide range of parametric models that is based on a Cram\'{e}r-von Mises distance between an unrestricted estimate of the joint cdf and the estimate of the joint cdf under the null hypothesis. However, the regressors are assumed to be constant for all quantiles. Thus, the $RW$ test approach equals $S_n^{CM}$ in case that the vector of transformations $P(X,\tau)$ is constant for all $\tau$.
\end{itemize} 
\begin{table}[htbp] \centering 
\captionof{table}{\normalsize{Size and power for RW/$S_n^{CM}$ and $S_n^{CM^*}$}}
\def\arraystretch{0.75}
\begin{tabular}{ p{2.cm}p{1.25cm}p{1.25cm} p{1.25cm} p{1.25cm} p{1.25cm} p{1.25cm} }
 \hline
\hline\\[-3.5ex]
   & \multicolumn{2} {c}{ DGP 1}&\multicolumn{2} {c}{DGP 2}&\multicolumn{2} {c}{DGP 3}\\
  \cmidrule(lr){2-3}\cmidrule(lr){4-5}\cmidrule(lr){6-7}
\rule{0pt}{1.5ex}     $RW / S_n^{CM}$& $10\%$ &$5\%$ &$10\%$ &$5\%$ &$5\%$ &$Power$ \\
\hline\\[-3.5ex]
 $n=30$ 	&0.077  &0.019  &0.093  &0.039  & 0.005& $\textbf{0.032}$ \\
 $n=50$ 	&0.061  &0.016  &0.095  &0.038  & 0.016& $\textbf{0.045}$  \\
 $n=100$ 	&0.056  &0.024  &0.087  &0.033  & 0.024& $\textbf{0.075}$  \\
 $n=300$ 	&0.055  &0.028  &0.078  &0.032  & 0.026& $\textbf{0.312}$  \\
 $n=500$ 	&0.056  &0.016  &0.069  &0.029  & $\textbf{0.010}$& 0.486  \\
 $n=1000$ 	&0.043  &0.016  &0.069  &0.030  & 0.014& 0.883  \\
 $n=2000$ 	&0.064  &0.020  &0.066  &0.030  & 0.014&${1.000}$ \\
\rule{0pt}{3ex}  $S_n^{CM^*}$& $10\%$ &$5\%$ &$10\%$ &$5\%$ &$5\%$ &$Power$\\
\hline\\[-3.5ex]
 $n=30$ 	&0.101  &0.035  &0.089  &0.037  & 0.028 &$\mathbf{0.095}$\\
 $n=50$ 	&0.103  & 0.046 &0.074  &0.027  & 0.037&$\mathbf{0.147}$\\
 $n=100$ 	&0.094  &0.043  &0.112  &0.061  & 0.064&$\mathbf{0.407}$\\
 $n=300$ 	&0.090  &0.043  &0.159  &0.084  & 0.047&$\mathbf{0.988}$\\
 $n=500$ 	&0.086  &0.043  &0.111  &0.058  & $\textbf{0.050}$&${1.000}$\\
 $n=1000$ 	&0.095  &0.048  &0.095  &0.038  & 0.056&${1.000}$\\
 $n=2000$ 	&0.098  &0.049  &0.092  &0.042  & 0.044&${1.000}$\\
 \hline  \hline 
\end{tabular}
\caption*{MC Study. The table compares the test statistics $S_n^{CM^*}$ and $S_n^{CM}$ in terms of size (significance levels $10\%$ and $5\%$) and power (at a $5\%$ evel), where the latter coincides with test statistic of \cite{rothewied:2013} ($RW$) in case of quantile-independent covariates.  The last column named \textit{Power} shows the power analysis while the qf is assumed to follow a linear  LSS model under the null hypothesis.}
\label{tab1_qaunt}\vspace{-.5cm}
\end{table} 
\paragraph{Results}
Table \ref{tab1_qaunt} shows the comparison with $RW$ for all $n$ in the univariate DGPs 1--3 in terms of size and power of the statistics at $10\%,5\%$ levels and a $5\%$ level, respectively. We make three observations.
First, compared with $RW/S_n^{CM}$ our proposed testing procedure $S_n^{CM^*}$ consistently has better size properties. 
Second, the test $S_n^{CM^*}$  manages to maintain the size level when the structure of the error terms is highly heteroscedastic (cf.~$5\%$ column of DGP $3$ in Table \ref{tab1_qaunt}).
Last, the rejection rate for misspecified models (for DGP $3$ we are assuming a linear LSS model in the last column of Table \ref{tab1_qaunt})  in small samples ($n\leq 300$) is approximately three times higher than for the $RW$ test. 
\\
Table \ref{wiedkoenker}  illustrates the comparison with $KX$ for the DGPs 4--8 for $n=100,300$, whereby a LS model is assumed under the null hypothesis. Thus, the results of DGPs $4$ and $5$ reflect size properties, while DGPs 6--8 illustrate the power of $S_n^{CM^*}$ compared with the benchmark tests $RW/S_n^{CM}$ and $KX$ at significance levels $10\%$ and $5\%$ each. We again make three observations.  
First,  our test $S_n^{CM^*}$ holds the size for multivariate models (cf.~DGPs $4$ and $5$ in Table \ref{wiedkoenker}). 
Second, $KX$ has difficulties to detect misspecification when heteroscedasticity is present (cf.~DGP $6-8$ in Table \ref{wiedkoenker}). 
Third, $RW/S_n^{CM}$ usually detects misspecification. However, the rejection rates of the test $S_n^{CM^*}$ are clearly higher compared with those from $RW$ even in small samples (cf.~$n=100$ DGP $7$ of Table \ref{wiedkoenker}).
\begin{table}[htbp]
\centering
\caption{\normalsize{Size and power for DGPs 4--8}}
\def\arraystretch{0.75}
\begin{tabular}{ p{1.95cm}p{1.5cm}p{1.5cm}p{1.5cm}p{1.5cm}p{1.5cm}p{1.5cm} }
 \hline
\hline\\[-3.5ex]
& \multicolumn{2} {c}{ $RW / S_n^{CM}$}&\multicolumn{2} {c}{$KX$}&\multicolumn{2} {c}{$S_n^{CM^*}$}\\ 
\hline\\[-3.5ex]
$n=100$& $10\%$ &$5\%$ &$10\%$ &$5\%$ &$10\%$ &$5\%$ \\
\hline\\[-3.5ex]
 DGP 4 	& 0.093 & 0.048 & 0.067 & 0.035 & {0.122}& \textbf{0.068 }\\
 DGP 5 	& 0.085 & 0.033 & 0.069 & 0.037 & {0.114}& \textbf{0.065 } \\
 DGP 6 	& 0.829 & \textbf{0.669} & 0.082 & \textbf{0.047} & {0.870}& \textbf{0.838 } \\
 DGP 7 	& 0.404 & \textbf{0.239} & 0.097 & \textbf{0.049} & {0.669}& \textbf{0.565 } \\
 DGP 8 	& 0.874 & \textbf{0.746} & 0.055 & \textbf{0.027} & {0.970}& \textbf{0.944 } \\
  \hline\\[-3.5ex]
 $n=300$& $10\%$ &$5\%$ &$10\%$ &$5\%$ &$10\%$ &$5\%$ \\
\hline\\[-3.5ex]
 DGP 4 	& 0.109 & 0.056 & 0.107 & 0.039 & {0.125}&\textbf{0.068 }\\
 DGP 5 	& 0.096 & 0.043 & 0.066 & 0.024 & {0.120}&\textbf{0.056}\\
 DGP 6 	& 1.000 & 0.997 & 0.336 & 0.231 & {1.000}&{1.000}\\
 DGP 7 	& 0.847 & 0.679 & 0.147 & 0.076 & {0.950}&{0.908 }\\
 DGP 8 	& 1.000 & 0.997 & 0.099 & 0.050 & {1.000}&{1.000}\\
 \hline 
 \hline
\end{tabular}
\caption*{MC Study. The table compares size and power (at significance level $5\%$) of the test statistics $RW/S_n^{CM}$, $KX$ and $S_n^{CM^*}$.  All results are one-to-one transferred from \citet{rothewied:2013}. The results of DGPs $4$ and $5$ reflect size properties, while DGPs 6--8 illustrate the power of $S_n^{CM^*}$ compared with the benchmark tests $RW/S_n^{CM}$ and $KX$ at significance levels $10\%$ and $5\%$ each.}
\label{wiedkoenker}\vspace{-.5cm}
\end{table}
Table \ref{kxcher} provides a comparison with the standard testing procedure proposed in \citet{koenkerxiao} and the enhancement from \citet{cher:2002} using $n=100,200,300$ and DGP 9. Results of Table $\ref{kxcher}$ of the benchmark tests $KX$ and $CH$ are taken from \citet{cher:2002}. From this table we conclude: The test $S_n^{CM^*}$ has consistently better finite sample properties compared with the benchmarks $KX$ and $CH$.
\begin{table}[htbp]
\centering
\caption{\normalsize{Size and power for DGP 9}}
\def\arraystretch{0.75}
\begin{tabular}{ p{1.5cm} p{1.1cm} p{1.1cm} p{1.1cm} p{1.1cm} p{1.1cm} p{1.1cm} p{1.1cm} p{1.1cm} p{1.1cm} }
 \hline\hline  \\[-3.5ex]
  & \multicolumn{3} {c}{ $KX$ }&\multicolumn{3} {c}{$CH$}&\multicolumn{3}{c}{$S_n^{CM^*}$}\\ 
 \cmidrule(lr){2-4}\cmidrule(lr){5-7}\cmidrule(lr){8-10}
\rule{0pt}{1.5ex} & \text{Size} &\multicolumn{2} {c}{ \text{Power}}&\text{Size} &\multicolumn{2} {c}{ \text{Power}}&\text{Size} &\multicolumn{2} {c}{ \text{Power}} \\
 \cmidrule(lr){2-2}\cmidrule(lr){3-4} \cmidrule(lr){5-5}\cmidrule(lr){6-7} \cmidrule(lr){8-8}\cmidrule(lr){9-10}
 $\gamma_1=$  & $0$ &$0.2$& $0.5$   & $0$ &$0.2$& $0.5$   & $0$ &$0.2$& $0.5$\\  
\hline\\[-3.5ex]
 $n=100$ &0.101 &0.264 & 0.898 & 0.014 & 0.348 & 0.980 &\textbf{0.050} &\textbf{0.396} &\textbf{0.99} 	\\
 $n=200$ &0.070 &0.480 & 0.988 & 0.052 & 0.752 & 1.000 &0.063&{0.772} &{1.000} 		\\
 $n=300$ &0.062 &0.622 & 0.998 & 0.058 & 0.910 & 1.000 &0.068& {0.930} &{1.000}		\\
  \hline
    \hline
\end{tabular}
\caption*{MC Study. The table compares size and power (at significance level 5\%) of the test statistics $KX$, $CH$ and $S_n^{CM^*}$. $KX$ refers to the specification test suggested by \cite{koenkerxiao}. The more powerful test of \cite{cher:2002} is abbreviated by $CH$.  All results are one-to-one transferred from \citet{cher:2002}. The null hypothesis assumes a LS quantile regression model, i.e.~$\gamma_1=0$.  }
\label{kxcher}\vspace{-.5cm}
\end{table}
Finally, Table \ref{tab0_quant} examines size and power properties for the DGPSs 10--12. Here, in each of the DGPs 10--12, the test $S_n^{CM^*}$ holds the significance level.  Assuming a linear model, misspecification is detected even in small sample sizes. DGP $10$ cannot be tested with previous approaches due to the quantile-dependent regressors. The slightly lower power for DGP 10 is due to the fact that half of the observations actually follow a linear relationship and are thus in line with the null hypothesis.

\begin{table}[htbp] \centering 
\captionof{table}{\normalsize{Size and power for DGPs 10--12}}
\def\arraystretch{0.75}
\begin{tabular}{ p{2.cm}p{1.5cm}p{1.5cm} p{1.5cm} p{1.5cm}p{1.5cm} p{1.5cm} }
 \hline
\hline\\[-3.5ex]
   & \multicolumn{2} {c}{ DGP 10}&\multicolumn{2} {c}{DGP 11} &\multicolumn{2} {c}{DGP 12} \\ 
  \cmidrule(lr){2-3}\cmidrule(lr){4-5} \cmidrule(lr){6-7} 
\rule{0pt}{1.5ex}     $S_n^{CM^*}$& $5\%$ &$Power$&$5\%$ &$Power$&$5\%$ &$Power$\\
\hline\\[-3.5ex]
 $n=30$ 	& 0.068 & 0.177 & 0.014 & 0.055 & 0.009 & 0.069\\
 $n=50$ 	& 0.057 & 0.189 & 0.018 & 0.285 & 0.013 & 0.318\\
 $n=100$ 	& \textbf{0.051} & \textbf{0.192} & \textbf{0.033} & \textbf{0.979} & \textbf{0.023} & \textbf{0.989}\\
 $n=300$ 	& 0.039 & 0.469 & 0.040 & 1.000 & 0.031 & 1.000\\
 $n=500$ 	& 0.042 & 0.519 & 0.039 & 1.000 & 0.029 & 1.000\\
 $n=1000$ 	& 0.046 & 0.658 & 0.034 & 1.000 & 0.035 & 1.000\\
 $n=2000$ 	& 0.042 & 0.743 & 0.041 & 1.000 & 0.049 & 1.000\\
 \hline  \hline 
\end{tabular}
\caption* {MC Study. The table reports size and power  of the test statistic $S_n^{CM^*}$ at a significance level $5\%$. }
\label{tab0_quant}\vspace{-.5cm}
\end{table} 

\subsection{MC Simulation Study for \texorpdfstring{$S_n^{CM, S}$}{SCMS}}\label{sec:MCsemi}
In this section, we show that our test $S_n^{CM, S}$ holds the size level and has considerable power properties using five different functional forms of complex interacting covariate effects, which we denote by DGP 13--17. Motivated by our second application, we use more flexible product tensor B-splines instead of product interactions  using univariate B-splines to model the interacting covariate effects. We consider this approach reasonable for two reasons. First, our second application reveals that the covariates indeed interact in a very complex way. Second, we show empirically in this section that multivariate tensor product B-splines yield satisfactorily testing results. 
For this, we couple DGPs 13--17 with various model specifications, which we refer to as B1--B6 (for two covariates) and T1--T5 (for three covariates). Overall, these settings attempt to mimic the situation of interaction effects as seen in our second real data illustration on electricity prices in Section~\ref{sec:edata} of our manuscript and are based on tensor product B-splines. For a detailed MC simulation study investigating univariate product interacting covariate effects, we refer to \ref{app:mcstudy}. Here, we demonstrate for two-dimensional functions that our test $S_n^CM^*$ holds the size level and has good power properties in case of product interacting covariate effects.

\paragraph{Data Generating Processes} DGPs 13--17 contain two or three interacting covariates and are defined as follows:
\begin{align}\label{Bdgp23}
\text{DGP 13:}\quad  &f_{13}(x_1,x_2):=7\cdot\sin(x_1\cdot x_2)\nonumber+x_1\\
\text{DGP 14:}\quad &f_{14}(x_1,x_2):=\sin(x_1\cdot x_2)+x_1\cdot x_2^2+z_1(x_1,x_2)\nonumber\\
\text{DGP 15:}\quad &f_{15}(x_3,x_4):=1+2x_3+4x_4+70\cos(x_3\cdot x_4)+u\\
\text{DGP 16:}\quad &f_{16}(x_2,x_5,u):=x_2^2\cdot x_5+x_2\cdot u+\cos(x_2 u)+z_1(x_2,x_5)\nonumber\\
\text{DGP 17:}\quad &f_{17}(x_1,x_2,x_3):=x_1+\sin(x_2)\cdot x_3+z_2(x_1,x_2,x_3)\nonumber
\end{align}
\\
Above, let $x_1\sim U(-4,-4)$, $x_2\sim N(5,1)$,  $x_3,x_4\sim U(0,1)$, $x_5\sim U(-10,10)$,  $z_1(x_1,x_2)\sim SN(x_1+x_2^2, 2+\sin(2x_1), x_1/4)$, $z_2(x_1,x_2,x_3)\sim SN(x_2+x_3^2, 5+\sin(x_1)x_3, x_3)$, $u\sim N(0,1)$, where $U(\cdot,\dot)$, $N(\cdot,\cdot)$ and $SN(\cdot,\cdot,\cdot)$ denote uniform, Gaussian and  skew normal distributions, respectively.  To estimate the functional forms in DGPs 13--17, we use cubic $P$-splines.

\paragraph{Model Specifications}
To increase the readability, the notation is geared to the implementation in R, i.e.~$s(\cdot, \tau)$ models the marginal $P$-spline  and $ti(\cdot,\cdot,\tau)$ for the interaction effect at the quantile $\tau$ excluding the basis functions associated with the lower dimensional marginal effects of the marginal smooths. For the case of two covariates, we define the following specifications:
\begin{align*}
\allowdisplaybreaks[4]
\begin{split}
\text{B1:}\qquad  &F_{Y\mid X}^{-1}(\tau\mid x_1,x_2):=s(x_1,\tau)\\
\text{B2:}\qquad &F_{Y\mid X}^{-1}(\tau\mid x_1,x_2):=s(x_1,\tau)+s(x_2,\tau)\\
\text{B3:}\qquad  &F_{Y\mid X}^{-1}(\tau\mid x_1,x_2):=s(x_1,\tau)+s(x_2,\tau)+ti(x_1,x_2,\tau)\\
\text{B4:}\qquad  &F_{Y\mid X}^{-1}(\tau\mid x_1,x_2):=\begin{cases}
s(x_1,\tau)+s(x_2,\tau),& \text{ if  } 0.25<\tau< 0.75\\
s(x_1,\tau)+s(x_2,\tau)+ti(x_1,x_2,\tau),& \text{ otherwise}
\end{cases}\\
\text{B5:}\qquad  &F_{Y\mid X}^{-1}(\tau\mid x_1,x_2):=\begin{cases}
s(x_1,\tau)+s(x_2,\tau),& \text{ if  } 0.25<\tau\\
s(x_1,\tau)+s(x_2,\tau)+ti(x_1,x_2,\tau),& \text{ otherwise}
\end{cases}\\
\text{B6:}\qquad &F_{Y\mid X}^{-1}(\tau\mid x_1,x_2):=\begin{cases}
s(x_1,\tau)+s(x_2,\tau),& \text{ if  } \tau< 0.75\\
s(x_1,\tau)+s(x_2,\tau)+ti(x_1,x_2,\tau),& \text{ otherwise}
\end{cases}\\
\end{split}
\end{align*}
\\
For the case of three covariates, let $s(x,\tau):=s(x_1,\tau)+s(x_2,\tau)+s(x_3,\tau)$. We define the following specifications:
\begin{align*} 
\allowdisplaybreaks[4]
\begin{aligned}
\text{T1:} &\quad F_{Y\mid X}^{-1}(\tau\mid x_1,x_2,x_3):=s(x,\tau)\\
\text{T2}:& \quad F_{Y\mid X}^{-1}(\tau\mid x_1,x_2,x_3):=s(x,\tau)+ti(x_1,x_2,\tau)+ti(x_2,x_3,\tau)\\
\text{T3}:& \quad F_{Y\mid X}^{-1}(\tau\mid x_1,x_2,x_3):=\begin{cases}
s(x,\tau)\ , &\hspace{-1.25cm}\text{if  } 0.25<\tau< 0.75\\
s(x,\tau)+ti(x_1,x_2,\tau)+ti(x_2,x_3,\tau)\ , &\text{otherwise}
\end{cases}\\
\text{T4}:& \quad F_{Y\mid X}^{-1}(\tau\mid x_1,x_2,x_3):=\begin{cases}
s(x,\tau)\ , &\text{if  } 0.25<\tau\\
s(x,\tau)+ti(x_1,x_2,\tau)+ti(x_2,x_3,\tau)\ , &\text{otherwise}
\end{cases}\\
\text{T5}:& \quad F_{Y\mid X}^{-1}(\tau\mid x_1,x_2,x_3):=\begin{cases}
s(x,\tau)\ , &\text{if  } \tau< 0.75\\
s(x,\tau)+ti(x_1,x_2,\tau)+ti(x_2,x_3,\tau)\ , &\text{otherwise}
\end{cases}\\
\end{aligned}
\end{align*}

\paragraph{Estimation and Further Settings} The estimation is carried out in the R-package \texttt{qgam} by \citet{fasiolo2020fast}. To keep the computational costs for our MC simulation study in reasonable limits and to make the results comparable, we use cubic $P$-splines with second order difference penalty and set the number of knots to $5$. We evaluate the test statistic for the quantiles $u\in\{0.02,0.04,\ldots,0.96,0.98\}$.
The number of overall replications is equal to $301$ and the significance level is set to $0.05$.

\begin{table}[htbp]
\centering
\caption{\normalsize{$S_n^{CM,S}$: Size and power for DGPs 13--15 and B1--B6}}
\footnotesize
\def\arraystretch{0.75}
\begin{tabular}{ p{1.cm}p{.8cm}p{.8cm}p{.8cm}p{.8cm}p{.8cm}p{.8cm}p{.8cm}p{.8cm}p{.8cm}p{.8cm}p{.8cm}p{.8cm}}
 \hline
\hline\\[-3.5ex]
& \multicolumn{6} {c}{ DGP 13}&\multicolumn{6} {c}{ DGP 14}  \\
  \cmidrule(lr){2-7}\cmidrule(lr){8-13} \\[-3.5ex]
&\text{B1}& \text{B2}& \textbf{B3} & \textbf{B4}  & \text{B5}  & \text{B6}  &\text{B1}& \text{B2}& \textbf{B3} & \textbf{B4}  & \text{B5}  & \text{B6}  \\ 
 \cmidrule(lr){2-2} \cmidrule(lr){3-3}\cmidrule(lr){4-4}\cmidrule(lr){5-5}\cmidrule(lr){6-6}\cmidrule(lr){7-7}\cmidrule(lr){8-8}\cmidrule(lr){9-9}\cmidrule(lr){10-10}\cmidrule(lr){11-11}\cmidrule(lr){12-12}\cmidrule(lr){13-13}
\text{n=500} & 1.000& 0.050& \textbf{0.033}& \textbf{0.033}& 0.058& 0.033& 1.000& 0.106& \textbf{0.057}& \textbf{0.076}& 0.0764& 0.089\\[1.5ex]
\text{n=1000}& 1.000& 0.088& \textbf{0.058}& \textbf{0.050}& 0.045& 0.067& 1.000& 0.103& \textbf{0.043}& \textbf{0.040}& 0.0565& 0.057\\[1.5ex]
\text{n=2000}& 1.000& 0.150& \textbf{0.079}& \textbf{0.046}& 0.070& 0.121& 1.000& 0.237& \textbf{0.050}& \textbf{0.057}& 0.1296& 0.156\\[1.5ex]
\text{n=3000}& 1.000& 0.392& \textbf{0.083}& \textbf{0.096}& 0.187& 0.292& 1.000& 0.445& \textbf{0.060}& \textbf{0.073}& 0.1761& 0.199\\[1.5ex]
\text{n=5000}& 1.000& 0.655& \textbf{0.046}& \textbf{0.067}& 0.241& 0.492& 1.000& 0.694& \textbf{0.066}& \textbf{0.080}& 0.2126& 0.269\\[1.5ex]
\text{n=6000}& 1.000& 0.867& \textbf{0.076}& \textbf{0.053}& 0.279& 0.613& 1.000& 0.764& \textbf{0.089}& \textbf{0.050}& 0.1927& 0.316\\[.5ex]
\hline\\[-2.5ex]
& \multicolumn{6} {c}{ DGP 15}&\multicolumn{6} {c}{Electricity Data}  \\
  \cmidrule(lr){2-7}\cmidrule(lr){8-13} \\[-3.5ex]
&\text{B1}& \text{B2}& \textbf{B3} & {B4}  & \text{B5}  & \text{B6}  &\text{B1}& \text{B2}& \textbf{B3} & {B4}  & \text{B5}  & \text{B6}  \\ 
 \cmidrule(lr){2-2} \cmidrule(lr){3-3}\cmidrule(lr){4-4}\cmidrule(lr){5-5}\cmidrule(lr){6-6}\cmidrule(lr){7-7}\cmidrule(lr){8-8}\cmidrule(lr){9-9}\cmidrule(lr){10-10}\cmidrule(lr){11-11}\cmidrule(lr){12-12}\cmidrule(lr){13-13}
\text{n=500} & 1.000& 1.000& \textbf{0.027}& 1.000& 1.000& 1.000& 0.661& 0.110 &\textbf{0.086}& 0.096& 0.126& 0.086\\[1.5ex]
\text{n=1000}& 1.000& 1.000& \textbf{0.019}& 1.000& 1.000& 1.000& 0.924& 0.099 &\textbf{0.059}& 0.076& 0.086& 0.077\\[1.5ex]
\text{n=2000}& 1.000& 1.000& \textbf{0.039}& 1.000& 1.000& 1.000& 1.000& 0.199 &\textbf{0.063}& 0.239& 0.206& 0.226\\[1.5ex]
\text{n=3000}& 1.000& 1.000& \textbf{0.019}& 1.000& 1.000& 1.000& 1.000& 0.329 &\textbf{0.057}& 0.435& 0.336& 0.422\\[1.5ex]
\text{n=5000}& 1.000& 1.000& \textbf{0.029}& 1.000& 1.000& 1.000& 1.000& 0.688 &\textbf{0.040}& 0.824& 0.688& 0.804\\[1.5ex]
\text{n=6000}& 1.000& 1.000& \textbf{0.049}& 1.000& 1.000& 1.000& 1.000& 0.864 &\textbf{0.057}& 0.917& 0.794& 0.920\\[1.5ex]
\hline
 \hline
\end{tabular}
\caption*{MC Study. Shown are  the size and power properties for the test statistic $S_n^{CM,S}$.  The columns with bold numbers depict the size of the specification test $S_n^{CM,S}$. The remaining columns represent the power of the test.}
\label{3d21}
\end{table} 

\paragraph{Results} Considering the bivariate case (cf.~Tables \ref{3d21}) we make four  observations. First, the test $S_n^{CM,S}$ holds the size level for the DGPs 13--15 (cf.~Tables \ref{3d21}, bold columns). Second, generally, power properties depend on the degree of misspecification. In the case of a moderately misspecified model (the difference between $B4$ and $B5$ or $B6$ is only that the lower 25\% or upper 75\% quantile contains interaction effects), the test $S_n^{CM,S}$ shows reasonable power properties. The higher the degree of misspecification the higher the rejection rates. This is particularly evident on DGP $15$, where the test always detects misspecification. Third, consistent with our theoretical investigations in Sec.~\ref{asymptotics}, the rejection rate for misspecified models increases with increasing sample size. Fourth, omitted variable bias is always detected (cf.~ Tables \ref{3d21}, column B1 in DGP 13--15).
\begin{table}[htbp]
\centering
\caption{\normalsize{$S_n^{CM,S}$:  Size and power  for DGPs 16--17 and B2, T1--T5}}
\footnotesize
\def\arraystretch{0.75}
\begin{tabular}{ p{1.2cm}p{.8cm}p{.8cm}p{.8cm}p{.8cm}p{.8cm}p{.8cm}p{.8cm}p{.8cm}p{.8cm}p{.8cm}p{.8cm}p{.8cm}}
 \hline
\hline\\[-3.5ex]
& \multicolumn{6} {c}{ DGP 16}&\multicolumn{6} {c}{DGP 17}  \\
  \cmidrule(lr){2-7}\cmidrule(lr){8-13} \\[-3.5ex]
&\text{B2}& \text{T1}& \textbf{T2} & \textbf{T3}  & \text{T4}  & \text{T5}  &\text{B2}& \text{T1}& \textbf{T2} & \textbf{T3}  & \text{T4}  & \text{T5}  \\ 
 \cmidrule(lr){2-2} \cmidrule(lr){3-3}\cmidrule(lr){4-4}\cmidrule(lr){5-5}\cmidrule(lr){6-6}\cmidrule(lr){7-7}\cmidrule(lr){8-8}\cmidrule(lr){9-9}\cmidrule(lr){10-10}\cmidrule(lr){11-11}\cmidrule(lr){12-12}\cmidrule(lr){13-13}
\text{$n=500$}  &0.476 &0.179 &\textbf{0.034} &\textbf{0.037} &0.186 &0.033&0.691& 0.073& \textbf{0.063}& \textbf{0.073}& 0.060& 0.076\\[1.5ex]
\text{$n=1000$} &0.754 &0.332 &\textbf{0.017} &\textbf{0.043} &0.303 &0.040&0.940& 0.083& \textbf{0.073}& \textbf{0.073}& 0.083& 0.073\\[1.5ex]
\text{$n=2000$} &0.898 &0.472 &\textbf{0.008} &\textbf{0.055} &0.458 &0.063&0.998& 0.083& \textbf{0.043}& \textbf{0.069}& 0.080& 0.070\\[1.5ex]
\text{$n=3000$} &0.984 &0.780 &\textbf{0.032} &\textbf{0.055} &0.764 &0.063&1.000& 0.149& \textbf{0.063}& \textbf{0.099}& 0.163& 0.123\\[1.5ex]
\text{$n=5000$} &1.000 &0.852 &\textbf{0.012} &\textbf{0.066} &0.835 &0.070&1.000& 0.179& \textbf{0.069}& \textbf{0.093}& 0.163& 0.096\\[1.5ex]
\text{$n=6000$} &1.000 &0.878 &\textbf{0.017} &\textbf{0.086} &0.889 &0.104&1.000& 0.183& \textbf{0.053}& \textbf{0.086}& 0.179& 0.089\\[1.5ex]
\hline
 \hline
\end{tabular}
\caption*{MC Study. Shown are the size and power properties for the test statistic $S_n^{CM,S}$.  The columns with bold numbers depict the size of the specification test $S_n^{CM,S}$. The remaining columns represent the power of the test.}
\label{3d3}\vspace{-.5cm}
\end{table} 
Considering the multivariate case with three covariates (cf.~DGPs 16--17 and B2, T1--T5  in Table \ref{3d3}), we make the following observations: 
First, the test $S_n^{CM,S}$ holds the size level (cf.~Table \ref{3d3}, bold columns).
Second, as in the bivariate cases, power properties depend on the degree of misspecification. For a moderately misspecified model ($T4$ and $T5$ contain interaction effects only in the lower 25\% or upper 75\% quantile), the test shows reasonable power properties. The higher the degree of misspecification the higher the rejection rates  (columns B2 and T1). Third, due to the curse of dimensions, however, the multivariate case with three covariates requires a larger number of observations $n$ to obtain similar properties as the multivariate case with two covariates.  Last, the omitted variable bias is sufficiently well detected for $n\geq2000$ (cf.~Table \ref{3d3}, column B2).

\subsection{Further Results from the MC Study using Product Interacting Covariates}\label{app:mcstudy}
In this section, we show that our test $S_n^{CM,S}$ holds the size level and has good power properties in case of product interacting covariates. For this, we consider the complex multivariate case with three covariates for the test statistic $S_n^{CM,S}$ (cf.~DGPs 16--17 and B2, T1--T5 in Table \ref{3d3_prod}), where we replace the tensor interaction in T2--T5 by univariate product interactions, i.e. we replace $ti$ by $s(x_1\cdot x_2)$, $s(x_2\cdot x_3)$ and $s(x_1\cdot x_2\cdot x_3)$, respectively. We make the following observations: 

\begin{table}[htbp]
\centering
\caption{\normalsize{Product interaction $S_n^{CM,S}$:  Size and power  for DGPs 16--17 and B2, T2--T5}}
\footnotesize
\def\arraystretch{0.75}%
\begin{tabular}{ p{1.2cm}p{.8cm}p{.8cm}p{.8cm}p{.8cm}p{.8cm}p{.8cm}p{.8cm}p{.8cm}p{.8cm}p{.8cm}p{.8cm}p{.8cm}}
 \hline
\hline\\[-3.5ex]
& \multicolumn{6} {c}{ DGP 16}&\multicolumn{6} {c}{DGP 17}  \\
  \cmidrule(lr){2-7}\cmidrule(lr){8-13} \\[-3.5ex]
&\text{B2}& \text{T1}& \textbf{T2} & \textbf{T3}  & \text{T4}  & \text{T5}  &\text{B2}& \text{T1}& \textbf{T2} & \textbf{T3}  & \text{T4}  & \text{T5}  \\ 
 \cmidrule(lr){2-2} \cmidrule(lr){3-3}\cmidrule(lr){4-4}\cmidrule(lr){5-5}\cmidrule(lr){6-6}\cmidrule(lr){7-7}\cmidrule(lr){8-8}\cmidrule(lr){9-9}\cmidrule(lr){10-10}\cmidrule(lr){11-11}\cmidrule(lr){12-12}\cmidrule(lr){13-13}
\text{$n=500$ }  &0.194& 0.043& \textbf{0.016}& \textbf{0.010}& 0.054& 0.010 &0.623& 0.041& \textbf{0.050}& \textbf{0.050}& 0.040& 0.054 \\[1.5ex]
\text{$n=1000$}  &0.461& 0.107& \textbf{0.026}& \textbf{0.029}& 0.131& 0.030 &0.884& 0.038& \textbf{0.024}& \textbf{0.030}& 0.031& 0.030 \\[1.5ex]
\text{$n=2000$}  &0.854& 0.273& \textbf{0.036}& \textbf{0.054}& 0.287& 0.060 &0.999& 0.059& \textbf{0.043}& \textbf{0.061}& 0.061& 0.059 \\[1.5ex]
\text{$n=3000$}  &0.984& 0.673& \textbf{0.039}& \textbf{0.069}& 0.666& 0.067 &1.000& 0.088& \textbf{0.063}& \textbf{0.087}& 0.091& 0.089 \\[1.5ex]
\text{$n=5000$}  &0.999& 0.900& \textbf{0.027}& \textbf{0.098}& 0.886& 0.126 &1.000& 0.128& \textbf{0.059}& \textbf{0.066}& 0.131& 0.067 \\[1.5ex]
\text{$n=6000$}  &1.000& 0.927& \textbf{0.033}& \textbf{0.087}& 0.924& 0.146 &1.000& 0.108& \textbf{0.054}& \textbf{0.063}& 0.113& 0.077 \\[1.5ex]
\hline
 \hline
\end{tabular}
\caption*{MC Study on product interactions. Shown are  the size and power properties for the test statistic $S_n^{CM,S}$ using product interaction effects, i.e. the tensor product splines in T2--T5 $ti(x_1,x_2), ti(x_2,x_3)$ are replace by product interactions $s(x_1\cdot x_2)$ and $s(x_2\cdot x_3)$, respectively. The columns with bold numbers depict the size of the specification test $S_n^{CM,S}$. The remaining columns represent the power of the test.}
\label{3d3_prod}\vspace{-.5cm}
\end{table} 

First, the test $S_n^{CM,S}$ holds the size level (cf.~Table \ref{3d3_prod}, bold columns). Second, as in the tensor product cases in Sec.~\ref{sec:MCsemi}, power properties depend on the degree of misspecification. For a moderately misspecified model ($T4$ and $T5$ contain interaction effects only in the lower 25\% or upper 75\% quantile), the test shows reasonable power properties. The higher the degree of misspecification the higher the rejection rates (cf. columns B2 and T1). Third, due to the curse of dimensions, however, the multivariate case with three covariates requires a large number of observations $n$ to obtain a powerful testing procedure as discussed in Sec.~\ref{sec:MCsemi}. Last, the omitted variable bias is sufficiently well detected for $n\geq2000$ (cf.~Table \ref{3d3_prod}, column B2).

\newpage
\section{Further Results from Modelling Australian Electricity Prices}\label{app:electricity}
\paragraph{Tensor product interacting covariates}
In addition to Sec.~\ref{sec:edata} of our manuscript, Figure \ref{fig_S6_Visuals_01} shows the decomposition of the main and interaction effects at the 10\% quantile at 6:00 p.m.~using specification S6. Since the contour lines in the second and third panel (upper right and lower left) show the presence of interactions between demand and day, we conclude that the relation between the three covariates cannot be fully captured by product interactions based on univariate splines. In addition, different day-demand combinations have a different impact on the market wide price $P_i$.

\begin{figure}[htbp]
	\centering
	\caption{\normalsize{Estimated main and product interaction effects at the 90\% quantile at 6:00 p.m.}}
	\includegraphics[width=1.\textwidth]{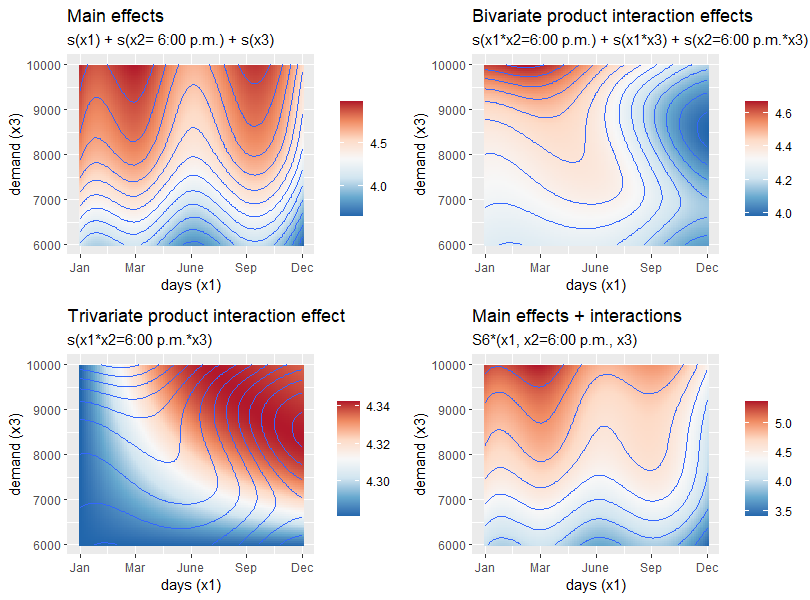}
	\caption*{Electricity prices. Figures depict the estimated effects of the three covariates on the 10\% quantile of the Australian NEM hourly electricity price distribution for 2019. The time of the day is set to 06:00 p.m. The estimation of the conditional qf is based on the sum of the main effects and bi- and trivariate product interaction effects, i.e. $F^{-1
	}_{Y\mid X}(\tau\mid x_1,x_2,x_3)=s(x_1,\tau)+s(x_2,\tau)+s(x_3,\tau)+s(x_1\cdot x_2,\tau)+s(x_1\cdot x_3,\tau)+s(x_2\cdot x_3,\tau)+s(x_1\cdot x_2\cdot x_3,\tau)$. The first panel (upper left) shows the sum of the univariate main effects of days (x1), time of day (x2) and total market demand (x3). The second and third panel illustrate the bivariate and trivariate product interaction effects. The overall effect is depicted in the last panel (lower right).}
	\label{fig_S6_Visuals_Prod}\vspace{-.5cm}
\end{figure}

\paragraph{Univariate product interacting covariates} Figure \ref{fig_S6_Visuals_Prod} shows the decomposition of the main and interaction effects at the 10\% quantile at 6:00 p.m.~using specification S6, where the tensor product interaction effects are replaced by univariate interacting covariates, i.e.~specification S6 is modified to 
\begin{align*}
\text{S6$^\ast$:}\qquad \qquad F^{-1
	}_{Y\mid X}(\tau\mid x_1,x_2,x_3)=s(x_1,\tau)&+s(x_2,\tau)+s(x_3,\tau)+s(x_1\cdot x_2,\tau)\\&+s(x_1\cdot x_3,\tau)+s(x_2\cdot x_3,\tau) +s(x_1\cdot x_2\cdot x_3,\tau).
\end{align*}
For the application of our test $S_n^{CM^*}$ with univariate product interacting effects to the Australian NEM, we consider a rolling window and set $n\in\{500,1000,2000\}$ and $\tau\in\{0.02,0.04,\ldots,0.98\}$. 
\begin{figure}[htbp]
	\centering
	\caption{\normalsize{Estimated main and product interaction effects at the 10\% quantile at 6:00 p.m.}}
	\includegraphics[width=1.\textwidth]{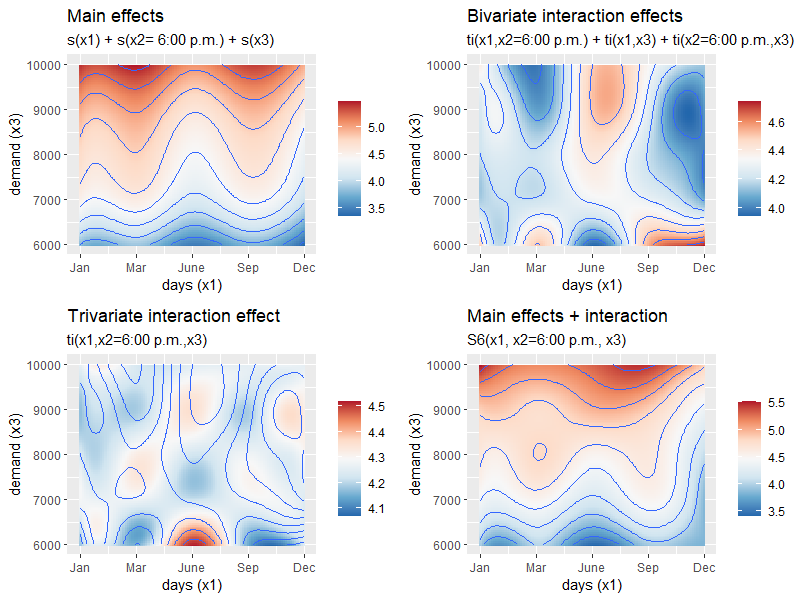}
	\caption*{Electricity prices. Figures depict the estimated effects of the three covariates on the 10\% quantile of the Australian NEM hourly electricity price distribution for 2019. The time of the day is set to 06:00 p.m.. The estimation is based on the model specification S6. The first panel (upper left) shows the sum of the univariate main effects of days (x1), time of day (x2) and total market demand (x3), where $x2$ is set to 6:00 p.m. The second and third panel illustrate the bivariate and trivariate interaction effects. The overall effect is depicted in the last panel (lower right).}
	\label{fig_S6_Visuals_01}\vspace{-.5cm}
\end{figure}
The number of sub-samples is $501$ and the critical values are calculated at a significance level of $5\%$. To ensure comparability of rejection rates for different $n$ and since we replaced the multivariate tensor by univariate interaction effects (cf.~S6 and S8--S13), we set the number of knots to $5$.  The rejection rates of the specification test $S_n^{CM*}$ are listed in Table \ref{tabstrom_prod}. For further details on the application we refer to Sec.~\ref{sec:edata}.
\begin{table}[htbp]
\centering
\caption{\normalsize Product Interaction: Rejection rates of the test statistic $S_n^{CM^*}$}
\footnotesize
\def\arraystretch{0.75}
\begin{tabular}{ p{1.1cm}cc ccc cc c}
 \hline
\hline\\[-3.5ex]
&\text{S6}& \text{S7}&\text{S8}&\text{S9}  & \text{S10}  & \text{S11} & \text{S12}  & \text{S13}    \\ 
\hline\\[-3.5ex]
\text{n=500}  &0.058& 0.072& 0.068& 0.058& 0.072& 0.056& 0.062& 0.081 \\[.5ex] 
\text{n=1000} &0.173& 0.235& 0.212& 0.183& 0.217& 0.164& 0.171& 0.192 \\[.5ex]
\text{n=2000} &0.646& 0.874& 0.773& 0.661& 0.739& 0.670& 0.685& 0.784\\[.5ex]
\hline
 \hline
\end{tabular}
\caption*{Electricity prices. The table shows the sub-sample rejection rates of size $n$ of the specification S6--S13, where the interaction effects modeled by tensor product splines $ti$ are replaced by product interaction effects, i.e. $s(x_1\cdot x_2),\ s(x_1\cdot x_3),\ s(x_2\cdot x_3)$ and $s(x_1\cdot x_2\cdot x_3)$.}
\label{tabstrom_prod}\vspace{-.5cm}
\end{table}

\end{document}